\documentclass[physrev,twocolumn,superscriptaddress,longbibliography,amsmath,amssymb,floatfix]{revtex4-2}

\usepackage{graphicx}
\usepackage{microtype}
\usepackage{hyperref}
\hypersetup{
    colorlinks = true,
    urlcolor   = blue,
    citecolor  = black,
    linkcolor  = black,
}

\newcommand{\dd}{\mathrm{d}}
\newcommand{\ii}{\mathrm{i}}

\newcommand{\avg}[1]{\langle #1 \rangle}
\newcommand{\norm}[1]{\| #1 \|}
\newcommand{\abs}[1]{\lvert #1 \rvert}
\newcommand{\vect}[1]{\boldsymbol{#1}}
\newcommand{\mat}[1]{\mathsf{#1}}

\newtheorem{lemma}{Lemma}
\newtheorem{theorem}{Theorem}
\newtheorem{proposition}{Proposition}
\newtheorem{corollary}{Corollary}

\newenvironment{proof}[1][Proof]{%
  \noindent\textit{#1.}\ }{%
  \hfill$\square$\par\medskip}

\begin{document}

\title{Unitary discretization of the Koopman--von Neumann equation 
       for quantum simulation of fluid and plasma dynamics}

\author{Aleksandar Jemcov}
\thanks{Corresponding author}
\email{ajemcov@nd.edu}
\affiliation{Department of Aerospace and Mechanical Engineering,
  University of Notre Dame, Notre Dame, Indiana 46556, USA}

\author{Scott C. Morris}
\email{smorris1@nd.edu}
\affiliation{Department of Aerospace and Mechanical Engineering,
  University of Notre Dame, Notre Dame, Indiana 46556, USA}

\begin{abstract}
The Koopman--von Neumann (KvN) formulation of spectrally truncated fluid and
plasma dynamics is considered as a potential approach for quantum computation.
The KvN framework embeds the Liouville equation
into a Hilbert space with norm-preserving, unitary evolution. Here we propose a Weyl-ordered KvN generator along with a summation-by-parts discretization which ensures that the resulting operators are exactly unitary as required for quantum computers. The Weyl-ordered KvN generator is derived as the unique
anti-Hermitian operator symmetrization for real velocity fields.
The formulation operates
directly in the physical amplitude space without phase-space doubling, so
the Heisenberg uncertainty principle does not constrain the grid resolution
during evolution. This limitation re-enters only at the measurement stage
on a quantum computer. Exact discrete unitarity is proved as a purely
algebraic identity that holds regardless of grid resolution or stencil
order. To manage boundaries, a split-step Kraus absorbing layer is
introduced via a Stinespring dilation requiring only one ancilla qubit.
Validation on three test cases spanning dissipative and Hamiltonian
regimes (a viscous Navier--Stokes triad, an
incompressible Euler triad, and a Hasegawa--Mima drift-wave triad)
confirms fourth-order convergence and machine-precision unitarity.
\end{abstract}

\maketitle

\section{Introduction}
\label{sec:intro}

Computing the statistical evolution of fluid and plasma flows directly
from their probability densities remains computationally intractable for
all but the simplest physical problems. Individual flow realizations follow, for example,
the Navier--Stokes or magnetohydrodynamic equations. However, when a flow
has many interacting degrees of freedom, the probability density function
(PDF) represents the preferred physical observable. Metrics
including statistical moments, energy spectra, and transport coefficients,
are all derived from this PDF. A single simulation run provides only one
sample from the underlying distribution. The Liouville equation serves as
the mathematical foundation for the exact evolution of this
density~\cite{Lasota1994}. Hopf~\cite{Hopf1952} derived an evolution
equation governing the characteristic functional of the velocity field,
but its infinite-dimensional nature has severely limited computational
applications. Contemporary approaches face well-known limitations. Moment
closure~\cite{Pope2000} introduces modeling errors. Monte Carlo (MC)
simulation~\cite{Fishman1996} converges at a rate of $1/\sqrt{N_s}$ for
$N_s$ samples. Direct PDF methods~\cite{Pope1985,Haworth2010} suffer from
exponential scaling of computational cost with dimension.

The Koopman--von Neumann (KvN) formulation~\cite{Koopman1931,vonNeumann1932,Mauro2002}
offers an alternative approach that offers the potential for implementation on quantum computing architectures. This approach embeds classical probability
densities into a Hilbert space with norm-preserving, unitary evolution,
yielding a linear representation of nonlinear dynamics that exactly
preserves probability. Joseph~\cite{Joseph2020} identified energy-conserving,
skew-symmetric finite differences~\cite{Morinishi1998,Morinishi2010} as the
appropriate discretization class for maintaining unitarity, and established
the theoretical quantum advantage. The $M^N$ phase-space amplitudes compress
into $N\lceil\log_2 M\rceil$ qubits, where $N$ is the number of retained
modes and $M$ is the number of grid points per modal dimension. This represents an
exponential reduction in state representation relative to classical
Eulerian discretization. Amplitude estimation~\cite{Brassard2002} then
provides a quadratic speedup for distributional queries when the KvN
operator is sparse. Joseph's formulation operates in a doubled
$(2N)$-dimensional phase space with explicit conjugate momenta. The
doubling is what makes the generator Hermitian and suitable for standard
Hamiltonian simulation, but it introduces Heisenberg uncertainty constraints
on joint phase-space resolution and leaves the handling of dissipation
(where unitarity breaks down) as an open problem.

Subsequent approaches address the non-unitarity that arises in
dissipative or open dynamics. Novikau and Joseph~\cite{Novikau2025}
applied the Linear Combination of Hamiltonian Simulations (LCHS)
technique to quantum simulation of both advection-diffusion and nonlinear
classical dynamics modeled via the KvN formulation. Their approach
decomposes the non-Hermitian propagator into a weighted sum of Hamiltonian
evolutions implemented through quantum signal processing, using upwind
discretizations that sacrifice Hermiticity to suppress numerical
oscillations; unitarity is recovered through the LCHS decomposition at
the cost of an additional register. A companion paper~\cite{Novikau2025b}
provides an efficient explicit implementation of this framework for linear
dissipative differential equations.
Jin, Liu, and Yu~\cite{Jin2023,Jin2024} proposed Schr\"{o}dingerization,
which lifts any $D$-dimensional non-unitary linear partial differential
equation (PDE) into a
$(D\!+\!1)$-dimensional Schr\"{o}dinger equation via a warped phase-space
transformation. The lifting requires an extra continuous dimension
discretized on $O(\log(1/\varepsilon))$ additional qubits. Other notable
approaches include Carleman linearization~\cite{Liu2021}, which truncates
an infinite-dimensional embedding, and quantum-walk and Koopman-operator
methods~\cite{Engel2019,Giannakis2022} that embed the classical dynamics
in quantum registers through different operator constructions but share
the limitation of low-degree polynomial dynamics.
Both LCHS and Schr\"{o}dingerization accept the non-unitarity introduced by
standard discretizations and compensate by adding auxiliary degrees of
freedom.

The present paper addresses several specific gaps left open by previous approaches.
For the real-valued PDEs discussed, we apply a modal decomposition of the physical domain variables to the phase space where the evolution of the amplitudes follows the Liouville equation. A new KvN generator is derived through operator promotion in a complex Hilbert space over the modal amplitude domain. We show that the Weyl-ordered structure is anti-Hermitian and is consistent with the Liouville equation. Second, phase-space discretization is accomplished using a Summation By Parts (SBP) operator that carries the anti-Hermitian property by algebraic identity. Hence, grid design and accuracy are determined only by classical resolution considerations.  Third, the problem of handling the phase space boundaries is addressed using a split-step Kraus absorbing layer. This  treats finite-domain truncation as a completely positive trace-preserving (CPTP) quantum channel, requiring only a single ancilla qubit via Stinespring dilation.

The formulation was validated on three spectrally truncated equations
spanning dissipative and Hamiltonian regimes: a viscous Navier--Stokes (NS)
triad~\cite{Waleffe1992,Holmes2012}, an
incompressible Euler triad~\cite{Craik1985,Waleffe1992}, and a
Hasegawa--Mima (HM) drift-wave triad~\cite{Hasegawa1978}. These three cases
span dissipative ($\nu > 0$) and Hamiltonian ($\nu = 0$) dynamics, and
contrast fluid ($1/|\vect{k}|^2$ Poisson inversion) with plasma
($1/(|\vect{k}|^2+1)$ polarization-modified inversion) coupling structure.

The remainder of this paper is organized as follows.
Section~\ref{sec:theory} derives the KvN generator
(Sec.~\ref{sec:kvn}), establishes exact discrete unitarity
(Theorem~\ref{thm:discrete_unitarity}, Sec.~\ref{sec:sbp}),
constructs the Kraus absorbing layer (Sec.~\ref{sec:kraus}), and proves
phase-space confinement bounds (Proposition~\ref{prop:norm_bound},
Sec.~\ref{sec:confinement}).
Section~\ref{sec:models} describes the three test cases and
identifies the coupling structure and conservation laws shared by all
three through their resonant triadic structure.
Section~\ref{sec:results} presents the numerical validation, including
grid convergence, domain sensitivity, and PDF evolution.
Section~\ref{sec:quantum} examines the implications for quantum
implementation.
Section~\ref{sec:conclusions} summarizes the findings.

\section{Theoretical formulation}
\label{sec:theory}

\subsection{Liouville--Koopman--von Neumann background}
\label{sec:kvn}

Koopman~\cite{Koopman1931} and von Neumann~\cite{vonNeumann1932}
showed that classical statistical mechanics can be formulated on Hilbert
space in a manner formally analogous to quantum mechanics. The classical
Liouville equation~\cite{Lasota1994}, which expresses the conservation
of probability on phase space, can be recast as an equivalent
Schr\"{o}dinger equation on the Hilbert space of square-integrable
wavefunctions on the modal amplitude domain~\cite{Mauro2002}. The formulation and notation used below follow this standard framework.

Consider a partial differential equation governing a fluid or plasma
flow, such as the incompressible Navier--Stokes, Euler, or
Hasegawa--Mima equation of the test cases in Sec.~\ref{sec:models}.
Galerkin projection onto a finite set of $N$ spatial basis functions,
typically Fourier modes fixed by the geometry and boundary conditions
of the parent PDE (see Appendix~\ref{app:parent_pdes}), reduces the
infinite-dimensional field equation to a coupled system of ordinary
differential equations (ODEs) for the modal amplitudes
$\vect{a}(t) = (a_1(t), \ldots, a_N(t)) \in \mathbb{R}^N$,
\begin{equation}
\dot{\vect{a}} = \vect{f}(\vect{a}),
\label{eq:ode}
\end{equation}
where $\vect{f}: \mathbb{R}^N \to \mathbb{R}^N$ encodes both the
nonlinear mode coupling inherited from the advection term and the
linear damping from viscous or collisional dissipation. For the
quadratic advection nonlinearities considered here,
$f_j = -\gamma_j a_j + \sum_{m,n} C_{jmn}\, a_m a_n$, with
$\gamma_j = \nu |\vect{k}_j|^2$ the viscous damping rate, $\nu$ the
kinematic viscosity, $\vect{k}_j$ the wavevector of the $j$th basis
mode, and $C_{jmn}$ the interaction coefficients set by the triad
geometry. The $N$-mode system~(\ref{eq:ode}) is obtained by projection alone. That is, no
moment closure is imposed, and the quadratic mode coupling is retained
exactly within the truncation. Explicit derivations of
$\vect{f}(\vect{a})$ for each test case are given in
Appendix~\ref{app:parent_pdes}.

A single trajectory of~(\ref{eq:ode}) is insufficient for the
statistical quantities of interest in turbulent and plasma flows.
Promoting $\vect{a}$ to a random variable distributed according to a
probability density $\rho(\vect{a},t)$ and requiring probability to be
conserved along the flow generated by $\vect{f}$ leads to the Liouville
equation~\cite{Lasota1994,Pope2000}, which is the continuity equation
for $\rho$ advected by the velocity field $\vect{f}$,
\begin{equation}
\frac{\partial \rho}{\partial t} + \nabla \cdot (\vect{f}\rho)
= \frac{\partial \rho}{\partial t} + \partial_j (f_j \rho) = 0,
\label{eq:liouville}
\end{equation}
where $\partial_j = \partial / \partial a_j$ and summation over repeated
$j$ is implied. The Liouville equation is linear in $\rho$ regardless
of whether $\vect{f}$ is linear, so the statistical evolution is
governed by a linear operator even when the underlying dynamics are
nonlinear. Applying
the Leibniz rule to the divergence,
\begin{equation}
\partial_j(f_j \rho) = f_j\, \partial_j \rho + (\partial_j f_j)\,\rho,
\label{eq:divergence_expansion}
\end{equation}
separates transport of $\rho$ along the flow from amplitude modulation
by the phase-space compressibility $\partial_j f_j$. The compressibility is the local rate at which
phase-space volumes contract or expand under the flow. It is negative
for viscous systems and identically zero for Hamiltonian systems. The
splitting~(\ref{eq:divergence_expansion}) is the structural feature
that couples the KvN generator to operator ordering in
Sec.~\ref{sec:weyl}, where the two terms correspond to the two non-commuting
orderings of the operators introduced below.

The Koopman--von Neumann formulation~\cite{Koopman1931,vonNeumann1932}
represents the statistical state as a complex
wavefunction $\psi(\vect{a},t)$ with
$\rho(\vect{a},t) = |\psi(\vect{a},t)|^2$. The wavefunction is a
probability amplitude. Its squared modulus is the classical probability
density, and $\psi$ itself belongs to the complex Hilbert space
$L^2(\Omega_a, \mathbb{C})$ of square-integrable functions on the modal
amplitude domain $\Omega_a \subset \mathbb{R}^N$, with the Euclidean
space $\mathbb{R}^N$ of modal amplitude vectors as the domain and
$\mathbb{C}$ signifying that $\psi$ takes complex values. The inner
product on $L^2(\Omega_a, \mathbb{C})$ is defined in the physics
convention, $\langle\varphi|\psi\rangle = \int_{\Omega_a}
\varphi^*(\vect{a})\psi(\vect{a})\,\dd\vect{a}$, conjugate-linear in
the first argument. Two norms will be used in what follows: the $L^2$
norm of the wavefunction,
$\norm{\psi} = (\int_{\Omega_a} |\psi|^2 \,\dd\vect{a})^{1/2}$, equal
to the total probability, and the Euclidean norm on $\mathbb{R}^N$, and 
$\norm{\vect{a}}_2 = (\sum_j a_j^2)^{1/2}$, equal to the length of a
modal state vector. The analytical advantage of the wavefunction
representation is that inner products, adjoints, and unitary evolution
become available on $L^2(\Omega_a, \mathbb{C})$; the implications for
a quantum-register implementation are addressed in
Sec.~\ref{sec:quantum}.

On this Hilbert space, the classical quantities of the Liouville
equation are promoted to operators. Multiplication by the component
$f_j(\vect{a})$ of the velocity field defines the operator
$\hat{f}_j : \psi \mapsto f_j(\vect{a})\,\psi$, and partial
differentiation with respect to $a_j$ defines
$\hat{\partial}_j : \psi \mapsto \partial_j \psi$. Both are unbounded
operators on $L^2(\Omega_a, \mathbb{C})$, with domains of smooth,
suitably decaying functions so that the integration-by-parts arguments
below are well defined. Their Hermitian adjoints follow from the
inner-product convention above. For the multiplication operator,
\begin{equation}
\langle\varphi | \hat{f}_j\,\psi\rangle
= \int_{\Omega_a} \varphi^* f_j\, \psi\,\dd\vect{a}
= \langle f_j^*\varphi | \psi\rangle,
\label{eq:mult_adjoint_main}
\end{equation}
so $\hat{f}_j^\dagger = \hat{f}_j^*$, and when $f_j$ is real-valued
$\hat{f}_j$ is self-adjoint. For the differentiation operator,
integration by parts with boundary contributions suppressed by the
decay of $\psi$ gives
\begin{multline}
\langle\varphi | \hat{\partial}_j\,\psi\rangle
= \int_{\Omega_a} \varphi^*\,\partial_j\psi\,\dd\vect{a} \\
= -\int_{\Omega_a} (\partial_j\varphi)^*\,\psi\,\dd\vect{a}
= -\langle\hat{\partial}_j\varphi | \psi\rangle,
\label{eq:diff_adjoint_main}
\end{multline}
so $\hat{\partial}_j^\dagger = -\hat{\partial}_j$, i.e.\ $\hat{\partial}_j$
is anti-self-adjoint. 
The operators
$\hat{f}_j$ and $\hat{\partial}_j$ do not commute, because $\partial_j$
acts on $f_j$ as well as on $\psi$; the precise form of the commutator
is needed in Sec.~\ref{sec:weyl} and is derived there. In terms of
these operators, the Liouville equation~(\ref{eq:liouville}) reads
$\partial_t\rho = -\hat{\partial}_j(\hat{f}_j\rho)$, and the
splitting~(\ref{eq:divergence_expansion}) is the statement that the
operator product $\hat{\partial}_j \hat{f}_j$ resolves into
$\hat{f}_j \hat{\partial}_j$ plus a correction proportional to
$\partial_j f_j$.

The wavefunction evolution is sought in the first-order linear form
\begin{equation}
\frac{\partial \psi}{\partial t} = \hat{L}\, \psi,
\label{eq:schrodinger}
\end{equation}
with $\hat{L}$ a linear operator on $L^2(\Omega_a, \mathbb{C})$
constructed from $\hat{f}_j$ and $\hat{\partial}_j$. The construction
is required to satisfy two conditions. First, the density $\rho =
|\psi|^2$ must evolve under~(\ref{eq:schrodinger}) in accordance with
the Liouville equation~(\ref{eq:liouville}). Computing
$\partial_t|\psi|^2 = (\hat{L}\psi)^*\psi + \psi^*\hat{L}\psi$ and
matching to $-\partial_j(f_j\rho)$ using the
splitting~(\ref{eq:divergence_expansion}) fixes $\hat{L}$ up to an
operator-ordering ambiguity that will be resolved in
Sec.~\ref{sec:weyl}. Second, probability must be conserved,
$\dd / \dd t \int \rho\,\dd\vect{a} = \dd / \dd t \norm{\psi}^2 = 0$.
Differentiating $\norm{\psi}^2 = \langle \psi | \psi \rangle$
under~(\ref{eq:schrodinger}) and using the adjoint identity
$\langle \hat{L}\psi | \psi \rangle = \langle \psi | \hat{L}^\dagger
\psi \rangle$ gives
\begin{equation}
\frac{\dd}{\dd t} \norm{\psi}^2
= \langle \psi | \hat{L}^\dagger \psi \rangle
+ \langle \psi | \hat{L}\psi \rangle
= \langle \psi | (\hat{L}^\dagger + \hat{L}) \psi \rangle,
\label{eq:norm_evolution}
\end{equation}
which vanishes for every $\psi$ if and only if the generator is
anti-Hermitian,
\begin{equation}
\hat{L}^\dagger = -\hat{L}.
\label{eq:antiherm_requirement}
\end{equation}
Anti-Hermiticity is thus a direct algebraic restatement of probability
conservation rather than an independent requirement. The KvN generator
plays the same role that $-\ii\hat{H}$ plays in the Schr\"{o}dinger
equation, with the sign flip reflecting the absence of the
conventional $\ii$ factor on the left-hand side
of~(\ref{eq:schrodinger}). Natural units ($\hbar = 1$) are adopted
throughout.

\subsection{Weyl ordering and the anti-Hermitian generator}
\label{sec:weyl}

Constructing an $\hat{L}$ that satisfies
Eq.~(\ref{eq:antiherm_requirement}) requires resolving an operator
ordering ambiguity. In the Hamiltonian case originally considered by
Koopman and von Neumann~\cite{Koopman1931,vonNeumann1932}, the
phase-space flow is divergence-free and the ordering is immaterial
because the two candidates differ by a term that vanishes identically.
For the compressible phase-space flows needed here,
$(\nabla\cdot\vect{f}\neq 0)$, the two
orderings are different and must be specified.
Weyl's symmetric prescription~\cite{Weyl1927} is the classical
resolution of such ambiguities in quantization; its use in KvN-type
discretizations has also been advocated on discrete
energy-conservation grounds~\cite{Joseph2020,Morinishi1998,Morinishi2010}.
The derivation below reaches the same prescription from a different
starting point, as a direct consequence of anti-Hermiticity on
$L^2(\Omega_a, \mathbb{C})$.

Applying the Leibniz rule to the operators $\hat{f}_j$ and
$\hat{\partial}_j$ introduced in Sec.~\ref{sec:kvn} gives
\begin{equation}
[\hat{\partial}_j,\, \hat{f}_j]\,\psi
\equiv \hat{\partial}_j(\hat{f}_j\psi) -
\hat{f}_j(\hat{\partial}_j\psi) = (\partial_j f_j)\,\psi,
\label{eq:commutator}
\end{equation}
where $[\hat{A}, \hat{B}] \equiv \hat{A}\hat{B} - \hat{B}\hat{A}$.
The two orderings $\hat{f}_j\hat{\partial}_j$ and
$\hat{\partial}_j\hat{f}_j$ therefore differ by the $j$th component of
the phase-space compressibility $\partial_j f_j$, which is the
operator form of the splitting~(\ref{eq:divergence_expansion}). Neither
ordering alone is anti-Hermitian. Using $\hat{f}_j^\dagger = \hat{f}_j$
(real $f_j$) and $\hat{\partial}_j^\dagger = -\hat{\partial}_j$, product
adjoints give $(\hat{f}_j\hat{\partial}_j)^\dagger = -\hat{\partial}_j
\hat{f}_j \neq -\hat{f}_j\hat{\partial}_j$, and symmetrically for the
opposite ordering, so the asymmetric combinations are ruled out.

Consider the one-parameter family of symmetric combinations
\begin{equation}
\hat{L}_\alpha = -\sum_{j=1}^{N}\bigl[\alpha\,\hat{f}_j\hat{\partial}_j
+ (1-\alpha)\,\hat{\partial}_j\hat{f}_j\bigr].
\label{eq:one_param_family_main}
\end{equation}
Its adjoint is
\begin{equation}
\hat{L}_\alpha^\dagger = \sum_{j=1}^{N}\bigl[\alpha\,\hat{\partial}_j
\hat{f}_j + (1-\alpha)\,\hat{f}_j\hat{\partial}_j\bigr].
\label{eq:one_param_adjoint}
\end{equation}
The anti-Hermiticity requirement $\hat{L}_\alpha^\dagger = -\hat{L}_\alpha$
therefore reads
\begin{equation}
\alpha\,\hat{\partial}_j\hat{f}_j + (1-\alpha)\,\hat{f}_j\hat{\partial}_j
= \alpha\,\hat{f}_j\hat{\partial}_j + (1-\alpha)\,\hat{\partial}_j\hat{f}_j,
\label{eq:antiherm_condition_main}
\end{equation}
with summation over $j$ implied. Substituting $\hat{\partial}_j\hat{f}_j
= \hat{f}_j\hat{\partial}_j + (\partial_j f_j)$
from~(\ref{eq:commutator}), the $\hat{f}_j\hat{\partial}_j$
contributions cancel on both sides and the condition collapses to the
scalar statement
\begin{equation}
(2\alpha - 1)\,\nabla\!\cdot\!\vect{f} = 0.
\label{eq:alpha_condition}
\end{equation}
The unique solution is $\alpha =
1/2$ for vector fields with nonzero divergence.
The anti-Hermitian ordering in the compressible case is therefore the
symmetric Weyl~\cite{Weyl1927} combination:
\begin{equation}
\hat{L} = -\tfrac{1}{2}\sum_{j=1}^{N} \bigl(\hat{f}_j \hat{\partial}_j +
\hat{\partial}_j \hat{f}_j\bigr).
\label{eq:kvn_generator}
\end{equation}

\begin{lemma}[Anti-Hermiticity]
\label{lem:antiherm}
The generator $\hat{L}$ satisfies $\hat{L}^\dagger = -\hat{L}$ if and
only if every $f_j(\vect{a})$ is real-valued.
\end{lemma}

\begin{proof}
Computing the adjoint term by term, using $\hat{f}_j^\dagger = \hat{f}_j^*$
(Eq.~(\ref{eq:mult_adjoint_main})) and $\hat{\partial}_j^\dagger = -\hat{\partial}_j$
(Eq.~(\ref{eq:diff_adjoint_main})) from Sec.~\ref{sec:kvn}, gives
$\hat{L}^\dagger = \tfrac{1}{2}\sum_j (\hat{f}_j^* \hat{\partial}_j
+ \hat{\partial}_j \hat{f}_j^*)$. The adjoint equals $-\hat{L}$ if and only
if $f_j^* = f_j$ for all $j$.
\end{proof}

The real-valued condition is a direct consequence of the underlying
physics. Galerkin projection of any real PDE onto real basis functions
produces real coupling coefficients $C_{jmn}$ and real damping rates
$\gamma_j$, so $\vect{f}: \mathbb{R}^N \to \mathbb{R}^N$ strictly.

Expanding Eq.~(\ref{eq:kvn_generator}) using the commutator
relation~(\ref{eq:commutator}) yields the differential-operator form
of the generator:
\begin{align}
\hat{L}\psi &= -\sum_j f_j\,\partial_j\psi
              - \tfrac{1}{2}\Bigl(\sum_j \partial_j f_j\Bigr)\psi \nonumber\\
            &= -\vect{f}\cdot\nabla\psi - \tfrac{1}{2}(\nabla\cdot\vect{f})\,\psi.
\label{eq:kvn_explicit}
\end{align}
The divergence term $-\tfrac{1}{2}(\nabla \cdot \vect{f})$ emerges
from the Leibniz rule applied to the $\hat{\partial}_j\hat{f}_j$ term
in the Weyl-ordered generator. For
divergence-free flows the generator reduces to pure advection and the
term vanishes. For dissipative equations it provides the
amplitude-modulation term required for exact recovery of
Eq.~(\ref{eq:liouville}). Verifying this explicitly, $\rho = |\psi|^2
= \psi^*\psi$ and $\partial_t\psi = \hat{L}\psi$ together give
$\partial_t\rho = (\hat{L}\psi)^*\psi + \psi^*(\hat{L}\psi)$;
substituting~(\ref{eq:kvn_explicit}) for both terms yields
\begin{align}
\partial_t\rho
&= \bigl[-\vect{f}\cdot\nabla\psi^*
   - \tfrac{1}{2}(\nabla\cdot\vect{f})\psi^*\bigr]\psi \nonumber\\
&\quad + \psi^*\bigl[-\vect{f}\cdot\nabla\psi
   - \tfrac{1}{2}(\nabla\cdot\vect{f})\psi\bigr] \nonumber\\
&= -\vect{f}\cdot\nabla\rho - (\nabla\cdot\vect{f})\,\rho \nonumber\\
&= -\nabla\cdot(\vect{f}\,\rho),
\label{eq:consistency}
\end{align}
where $\psi\nabla\psi^* + \psi^*\nabla\psi = \nabla\rho$ collects the
gradient terms and the two half-divergence terms combine to the full
$(\nabla\cdot\vect{f})\,\rho$ required by~(\ref{eq:divergence_expansion}).
The Liouville equation~(\ref{eq:liouville}) is recovered exactly, which
confirms the first of the two conditions imposed in Sec.~\ref{sec:kvn}.

Because $\hat{L}$ has real coefficients and the initial wavefunction
satisfies $\psi(\vect{a},0) = \sqrt{\rho(\vect{a},0)} \geq 0$, the
wavefunction remains real for all $t > 0$. The complex Hilbert space
structure is nonetheless required for the formal anti-Hermiticity proof
and for the quantum mapping discussed in Sec.~\ref{sec:quantum}.

\subsection{Exact discrete unitarity}
\label{sec:sbp}

The continuous generator $\hat{L}$ acts on $L^2(\Omega_a, \mathbb{C})$.
For computation, the generator is discretized on a uniform grid with $M$ points per
coordinate direction and spacing $\Delta a = L_a/M$, where $L_a$ is the
extent of $\Omega_a$ in each coordinate direction. The wavefunction is represented by $K = M^N$ complex
amplitudes $\psi_{\vect{i}}$, indexed by the multi-index
$\vect{i} = (i_1, \ldots, i_N)$ with components
$i_k \in \{0, 1, \ldots, M-1\}$ for $k = 1, \ldots, N$. The inner product becomes
$\langle\varphi|\psi\rangle \approx \sum_{\vect{i}} \varphi_{\vect{i}}^*
\psi_{\vect{i}}\,\Delta V$ for any $\varphi, \psi \in L^2(\Omega_a,
\mathbb{C})$, where $\Delta V = (\Delta a)^N$.

Two matrices are needed for each coordinate direction $j$. The
\emph{differentiation matrix} $\mat{D}_j$ approximates
$\partial/\partial a_j$ on the discrete grid through the Kronecker
product
\begin{equation}
\mat{D}_j = \mat{I} \otimes \cdots \otimes \mat{D} \otimes \cdots
\otimes \mat{I},
\label{eq:Dj_kronecker}
\end{equation}
where the one-dimensional, fourth order of accuracy operator $\mat{D}^{(4)}$
(Appendix~\ref{app:sbp4}) appears in the $j$th position and $\mat{I}$
is the $M \times M$ identity in all other directions. For the
anti-Hermiticity of the continuous operator to carry over, $\mat{D}$
must be exactly skew-symmetric: $\mat{D}^{\mathrm{T}} = -\mat{D}$. This
requirement is fulfilled by periodic Summation-By-Parts (SBP)
operators~\cite{Kreiss1974,Strand1994,Svard2014}
(Appendix~\ref{app:sbp}), whose exact skew-symmetry is proved for the
specific operators used here in Appendix~\ref{app:sbp_skewsym}:
\begin{equation}
\mat{D}_j^\dagger = \mat{D}_j^{\mathrm{T}} = -\mat{D}_j
\quad (\text{anti-Hermitian}).
\label{eq:D_antiherm}
\end{equation}
The \emph{velocity-field matrix}
$\mat{F}_j = \mathrm{diag}(f_j(\vect{a}_{\vect{i}}))$ samples the
modal velocity field at every grid point. Since $f_j$ is always
real-valued (Lemma~\ref{lem:antiherm}), $\mat{F}_j$ is real diagonal
and self-adjoint:
\begin{equation}
\mat{F}_j^\dagger = \mat{F}_j^{\mathrm{T}} = \mat{F}_j
\quad (\text{Hermitian}).
\label{eq:F_herm}
\end{equation}
The Weyl-ordered discrete generator is
\begin{equation}
\mat{L} = -\frac{1}{2} \sum_{j=1}^N \left(
  \mat{F}_j \mat{D}_{j} + \mat{D}_{j} \mat{F}_j \right).
\label{eq:kvn_discrete}
\end{equation}

\begin{theorem}[Exact discrete unitarity]
\label{thm:discrete_unitarity}
Let $\mat{D}_j$ be a periodic SBP first-derivative operator of any
stencil order on a uniform grid, satisfying
$\mat{D}_j^{\mathrm{T}} = -\mat{D}_j$ (proved for the explicit operators
used here in Appendix~\ref{app:sbp_skewsym}), and let
$\mat{F}_j = \mathrm{diag}(f_j(\vect{a}_{\vect{i}}))$ be a real
diagonal matrix of velocity-field values, satisfying
$\mat{F}_j^{\mathrm{T}} = \mat{F}_j$. Then the Weyl-ordered discrete
generator~(\ref{eq:kvn_discrete}) satisfies
$\mat{L} + \mat{L}^{\mathrm{T}} = 0$ exactly, and the propagator
$\mat{U} = \exp(\mat{L}\,\Delta t)$ is exactly unitary,
$\mat{U}^{\mathrm{T}}\mat{U} = \mat{I}$, for any $\Delta t$, any
velocity field $\vect{f}$, any divergence structure
$\nabla\!\cdot\!\vect{f}$, and any grid resolution $M$.
\end{theorem}

\begin{proof}
Transposing~(\ref{eq:kvn_discrete}) and applying
$(\mat{A}\mat{B})^{\mathrm{T}} = \mat{B}^{\mathrm{T}}\mat{A}^{\mathrm{T}}$:
\begin{equation}
\mat{L}^{\mathrm{T}}
= -\frac{1}{2}\sum_{j=1}^{N}
  \bigl(\mat{D}_j^{\mathrm{T}}\mat{F}_j^{\mathrm{T}}
       +\mat{F}_j^{\mathrm{T}}\mat{D}_j^{\mathrm{T}}\bigr).
\end{equation}
Substituting $\mat{D}_j^{\mathrm{T}} = -\mat{D}_j$~(\ref{eq:D_antiherm})
and $\mat{F}_j^{\mathrm{T}} = \mat{F}_j$~(\ref{eq:F_herm}):
\begin{align}
\mat{L}^{\mathrm{T}}
&= \frac{1}{2}\sum_{j=1}^{N}
  \bigl(\mat{D}_j\mat{F}_j + \mat{F}_j\mat{D}_j\bigr)
= -\mat{L},
\end{align}
so $\mat{L} + \mat{L}^{\mathrm{T}} = 0$ exactly. This identity uses no
property of $\vect{f}$, $\Delta a$, $M$, or stencil order beyond the
two adjoint conditions above. For the propagator, since $\mat{L}$ is real:
\begin{align}
\mat{U}^{\mathrm{T}}\mat{U}
= \exp(-\mat{L}\,\Delta t)\exp(\mat{L}\,\Delta t)
= \mat{I}.
\end{align}
\end{proof}

The evolution $\psi^{n+1} = \mat{U}\,\psi^n$ is the exact time-stepping
form of~(\ref{eq:schrodinger}). Probability conservation
$\|\psi^{n+1}\| = \|\psi^n\|$ follows from the strict unitarity of
$\mat{U}$. The physics-driven nonlinearity in $\vect{f}$ is carried
exactly by the diagonal matrices $\mat{F}_j$, without mathematical
truncation or linearization. The resulting discrete generator coincides
with the skew-symmetric form
$\tfrac{1}{2}(\mat{F}\mat{D} + \mat{D}\mat{F})$ that
Morinishi~\cite{Morinishi1998,Morinishi2010} derived by rewriting
the convective terms to conserve discrete kinetic energy for
incompressible flows, and that Joseph~\cite{Joseph2020} identified as
the appropriate discretization class for maintaining KvN unitarity.
The present derivation reaches the same algebraic structure from the
Hilbert-space embedding. The Weyl ordering resolves the operator
ambiguity in $L^2(\Omega_a, \mathbb{C})$, anti-Hermiticity requires real
velocity fields (Lemma~\ref{lem:antiherm}), and SBP
operators~(\ref{eq:D_antiherm}) supply the discrete skew-symmetry. The
convergence of the two routes reflects that both enforce quadratic norm
conservation; the Hilbert-space route does not require
$\nabla \cdot \vect{f} = 0$ and extends to the quantum setting.

\subsection{Boundary treatment and phase-space confinement}
\label{sec:kraus}

The periodic SBP operator assumes a periodic phase-space domain. A finite
domain requires treatment at its edges to prevent probability that
approaches one boundary from reappearing at the opposite boundary as a
spurious wrap-around artifact. Rather than modifying the interior
operator, which would destroy the algebraic anti-Hermiticity, a split-step
scheme separates interior transport from boundary absorption. The Kraus
absorbing layer attenuates $\psi$ near each domain edge;
its CPTP structure provides exact probability accounting that classical
sponge-zone formulations lack. The boundary treated here is the edge of
the computational phase-space box $\Omega_a \subset \mathbb{R}^N$, not
the boundary of the physical-space domain on which the parent PDE is
posed. Physical-space boundary conditions enter only at the projection
stage, through the choice of Galerkin basis and the resulting coupling
coefficients $C_{jmn}$.

The layer was built from a one-dimensional cosine taper over a layer of
width $n_s$ grid spacings from each boundary:
\begin{equation}
w(d) = \begin{cases}
\tfrac{1}{2}\bigl(1 + \cos\bigl(\pi(1 - d/n_s)\bigr)\bigr), & d < n_s, \\[4pt]
1, & d \geq n_s,
\end{cases}
\label{eq:taper}
\end{equation}
where $d$ is the distance in grid spacings from the nearest boundary.
The full $N$-dimensional weight at grid point $\vect{i} = (i_1,\ldots,i_N)$
is the separable product
\begin{equation}
W_{\vect{i}} = \prod_{k=1}^{N} w\,\!\bigl(d(i_k)\bigr) \in [0,1],
\label{eq:weight}
\end{equation}
where $d(i_k)$ is the distance from the nearest boundary in direction $k$.
The weight equals unity in the interior and vanishes at the domain boundary.
The separable structure maps directly to single-register quantum gates.

Two Kraus operators are defined from this weight~(\ref{eq:weight}):
$\mat{K}_1 = \mathrm{diag}(\sqrt{W_{\vect{i}}})$, which retains
probability in the interior, and
$\mat{K}_2 = \mathrm{diag}(\sqrt{1-W_{\vect{i}}})$, which gives
the absorbed fraction. In the computational basis of the quantum register
these read $\mat{K}_1 = \sum_{\vect{i}}\sqrt{W_{\vect{i}}}|\vect{i}\rangle\langle\vect{i}|$
and $\mat{K}_2 = \sum_{\vect{i}}\sqrt{1-W_{\vect{i}}}|\vect{i}\rangle\langle\vect{i}|$,
making the support of each operator immediately transparent. The interior
propagator is applied first; then at each grid point,
\begin{equation}
\psi^{n+1}_{\vect{i}} = \sqrt{W_{\vect{i}}}\;(\mat{U}\,\psi^n)_{\vect{i}},
\label{eq:split_step}
\end{equation}
and the Kraus operators satisfy the completeness relation
\begin{equation}
\mat{K}_1^{\mathrm{T}} \mat{K}_1 + \mat{K}_2^{\mathrm{T}} \mat{K}_2
= \mat{I},
\label{eq:cptp}
\end{equation}
which is the discrete Kraus representation of a CPTP quantum
channel~\cite{Nielsen2010}. The retained probability and its
complement are defined by:
\begin{equation}
P_{\mathrm{ret}}(t) \equiv \norm{\psi(t)}^2
= \sum_{\vect{i}} \abs{\psi_{\vect{i}}(t)}^2 \Delta V,
\label{eq:pret_def}
\end{equation}
\begin{equation}
P_{\mathrm{lost}}(t) \equiv 1 - P_{\mathrm{ret}}(t).
\label{eq:plost_def}
\end{equation}
The completeness relation~(\ref{eq:cptp}) guarantees
$P_{\mathrm{ret}}(t) + P_{\mathrm{lost}}(t) = 1$ to machine precision
at every step. The quantity $P_{\mathrm{lost}}$ serves as an a posteriori
diagnostic. Its smallness confirms that the computational domain is
adequately sized.

The interior step is a unitary channel representing closed-system evolution;
the Kraus layer is a non-unitary channel representing interaction with an
environment. The non-unitary channel admits a unitary dilation on an
extended Hilbert space, the form required for execution on quantum
hardware. The dilation construction, ancilla-circuit realization, and
qubit and gate-count overheads are presented in
Sec.~\ref{sec:quantum}.

Whether this absorbing layer actively absorbs probability or instead
remains inactive throughout the simulation depends on whether the
classical dynamics confine trajectories to a bounded region of phase
space\label{sec:confinement}. The phase-space bounds can be established
analytically from the modal amplitude equations, independently of any
choice of grid or domain. Two scalar quantities built from the modal
state $\vect{a} = (a_1,\ldots,a_N)$ encode the relevant information:
\begin{equation}
Z(\vect{a}) \equiv \tfrac{1}{2}\norm{\vect{a}}_2^2
= \tfrac{1}{2}\sum_{j=1}^{N} a_j^2,
\label{eq:Z_def}
\end{equation}
the \emph{squared modal amplitude}, and
\begin{equation}
T(\vect{a}) \equiv \sum_{j,m,n=1}^{N} C_{jmn}\, a_j\, a_m\, a_n,
\label{eq:cubic_form}
\end{equation}
the \emph{cubic contraction of the interaction coefficients}. $Z$ is
half the squared Euclidean ($\ell^2$) norm of the modal state vector,
$\norm{\vect{a}}_2 = (\sum_j a_j^2)^{1/2}$. It is not the kinetic
energy (which carries the $|\vect{k}|^{-2}$ weighting specified
in Sec.~\ref{sec:models}), but a geometric measure of how far the state
vector lies from the origin of phase space. A bound on $Z(t)$ translates
directly into a bound on $\norm{\vect{a}(t)}_2$, and hence into the
required size of the computational domain $\Omega_a$.

All three test cases of Sec.~\ref{sec:models} fit the general $N$-mode
form with diagonal linear damping and quadratic nonlinearity,
\begin{equation}
\dot{a}_j = \sum_{m,n=1}^{N} C_{jmn}\, a_m a_n - \gamma_j\, a_j,
\qquad j = 1,\dots,N,
\label{eq:general_quadratic}
\end{equation}
with damping rates $\gamma_j \geq 0$ and slowest rate
$\gamma_{\min} \equiv \min_{1\leq j\leq N} \gamma_j$. Differentiation
of $Z$ along trajectories of~(\ref{eq:general_quadratic}), combined with
the definition~(\ref{eq:cubic_form}) of $T$, yields
\begin{equation}
\dot{Z} \;=\; \sum_{j=1}^N a_j\, \dot{a}_j
\;=\; T(\vect{a}) \;-\; \sum_{j=1}^N \gamma_j\, a_j^2.
\label{eq:Zdot}
\end{equation}

Identity~(\ref{eq:Zdot}) presents $Z$ as a candidate Lyapunov function
for the origin $\vect{a} = 0$. The quantity $Z$ is positive definite
and continuously differentiable. The dissipative term in $\dot Z$ is
sign-definite and contracting, and the nonlinear term is sign-indefinite,
the only mechanism through which $Z$ can grow. Lyapunov stability of
the origin therefore reduces to controlling the nonlinear contribution.

The structural condition $T(\vect{a}) \equiv 0$ supplies that control.
It has a direct physical interpretation. When $\gamma_j$ is set to zero
in~(\ref{eq:general_quadratic}), the time derivative of
$\norm{\vect{a}}_2^2$ along trajectories reduces to $2T(\vect{a})$, so
$T \equiv 0$ is precisely the statement that the nonlinearity, taken
alone, conserves the modal $\ell^2$ norm. The identity is the
modal-amplitude form of the conservation property carried by
skew-symmetric advection in the parent PDE, and is the same algebraic
norm preservation inherited at the Hilbert-space level by the
anti-Hermitian KvN generator
(Sec.~\ref{sec:weyl}) and at the discrete operator level by the SBP
propagator (Sec.~\ref{sec:sbp}). The conservative structure that
produces unitary evolution at the operator level produces a strict
Lyapunov function at the modal-amplitude level. The energy-budget identity~(\ref{eq:Zdot}) and its consequence under cyclic cancellation are classical~\cite{Lee1952,Kraichnan1959,Waleffe1992}, and the conversion to an exponential bound is the standard Lyapunov argument for quadratic systems with linear damping~\cite{Khalil2002}; what is specific to the present analysis is the use of the resulting bound to certify numerical inactivity of the Kraus absorbing layer.

The hypothesis $T \equiv 0$ is satisfied by Galerkin truncations
that inherit a quadratic conservation law from the parent PDE, and
fails when external forcing, off-resonant interactions, or non-cyclic
coupling coefficients break the inherited identity. The three test
cases of Sec.~\ref{sec:models} are all of the former type, by the
cyclic telescoping identity for resonant triads derived in
Appendix~\ref{app:parent_pdes}. In broader regimes where $T \not\equiv 0$,
confinement is no longer guaranteed by the analysis below and the
Kraus absorbing layer becomes the mechanism that maintains a well-posed
simulation on a finite domain.

When $T \equiv 0$ holds and $\gamma_{\min} > 0$, the damping
in~(\ref{eq:Zdot}) drives $Z$
strictly downward at a rate fixed by $\gamma_{\min}$, and the origin
becomes globally exponentially stable in $\norm{\vect{a}}_2$. The
following proposition states this conclusion quantitatively.

\begin{proposition}[Phase-space norm bound]
\label{prop:norm_bound}
For the system~(\ref{eq:general_quadratic}) with damping rates
$\gamma_j > 0$, let $\gamma_{\min} \equiv \min_{1\leq j\leq N}\gamma_j > 0$.
If the cubic contraction~(\ref{eq:cubic_form}) vanishes identically,
$T(\vect{a}) \equiv 0$, then for every initial condition
$\vect{a}(0) \in \mathbb{R}^N$ the solution exists for all $t \geq 0$
and satisfies
\begin{equation}
\norm{\vect{a}(t)}_2 \;\leq\; \norm{\vect{a}(0)}_2\, e^{-\gamma_{\min}\, t}.
\label{eq:norm_bound}
\end{equation}
In particular, $\norm{\vect{a}(t)}_2$ is monotonically non-increasing.
\end{proposition}

\begin{proof}
The polynomial right-hand side of~(\ref{eq:general_quadratic})
guarantees a unique smooth solution on a maximal interval $[0,T_*)$.
On this interval, the hypothesis $T \equiv 0$ reduces~(\ref{eq:Zdot})
to $\dot{Z} = -\sum_j \gamma_j\, a_j^2$. The slowest-damped rate bounds
this sum from below: $\gamma_j \geq \gamma_{\min}$ for every $j$ gives
\begin{equation*}
\sum_{j=1}^N \gamma_j\, a_j^2
\;\geq\; \gamma_{\min} \sum_{j=1}^N a_j^2.
\end{equation*}
The factor $\tfrac12$ in the definition $Z = \tfrac12 \norm{\vect{a}}_2^2$
inverts to give $\sum_j a_j^2 = 2Z$, so the bound becomes
$\dot{Z} \leq -2\gamma_{\min}\, Z$, a linear differential inequality.
The integrating-factor identity
\begin{equation*}
\frac{d}{dt}\!\left[Z(t)\, e^{2\gamma_{\min} t}\right]
\;=\; \bigl(\dot{Z} + 2\gamma_{\min}\, Z\bigr)\, e^{2\gamma_{\min} t}
\;\leq\; 0
\end{equation*}
yields $Z(t) \leq Z(0)\, e^{-2\gamma_{\min} t}$. Taking square roots
halves the rate exponent and recovers~(\ref{eq:norm_bound}) via
$Z = \tfrac12 \norm{\vect{a}}_2^2$. The bound is uniform on $[0,T_*)$
and precludes finite-time blow-up, so $T_* = \infty$.
\end{proof}

Fourier modes couple through the quadratic advection nonlinearity only
when their wavevectors satisfy a closure condition: the product
$e^{\ii\vect{k}_m\cdot\vect{x}}\,
e^{\ii\vect{k}_n\cdot\vect{x}} = e^{\ii(\vect{k}_m+\vect{k}_n)\cdot\vect{x}}$
projects onto a third mode $\vect{k}_j$ only when
$\vect{k}_j = -(\vect{k}_m+\vect{k}_n)$, so that
$\vect{k}_1 + \vect{k}_2 + \vect{k}_3 = \vect{0}$. For triadic truncations
satisfying this resonance condition, the nonzero components
of $C_{jmn}$ reduce to cyclic entries $C_j$, so
$T(\vect{a}) = a_1 a_2 a_3 \sum_j C_j$. The condition $T \equiv 0$ then
reduces to $\sum_j C_j = 0$. Throughout this paper, the notation
$(j,m,n)$ \emph{cyclic} denotes the three index triples obtained by
cyclic permutation of $(1,2,3)$: $(1,2,3)$, $(2,3,1)$, and $(3,1,2)$.
The three test cases all take the triadic form
\begin{equation}
\dot{a}_j = -\gamma_j a_j + C_j a_m a_n, \quad (j,m,n) \text{ cyclic};
\end{equation}
their coupling coefficients and conservation laws are recorded in
Sec.~\ref{sec:models}. For this form, the following corollary applies.

\begin{corollary}[Triadic confinement]
\label{cor:triad_confinement}
\begin{sloppypar}
For a resonant triad with $\gamma_j > 0$ satisfying $\sum_j C_j = 0$, the
cubic form $T(\vect{a}) = a_1 a_2 a_3 \sum_j C_j$ vanishes identically.
Proposition~\ref{prop:norm_bound} then implies that every trajectory
satisfies $\norm{\vect{a}(t)}_2 \leq \norm{\vect{a}(0)}_2$; the domain
need only contain the initial probability support.
\end{sloppypar}
\end{corollary}

The three test cases of Sec.~\ref{sec:models} realize two combinations
of the proposition's hypotheses, and the role of the Kraus absorbing
layer differs accordingly. The Navier--Stokes and Hasegawa--Mima triads
satisfy both $T \equiv 0$ and $\gamma_{\min} > 0$, so
Proposition~\ref{prop:norm_bound} applies directly: the modal norm
decays exponentially at rate $\gamma_{\min}$, trajectories spiral
toward the origin, and the support of the wavefunction $\psi$ is
actively contracted away from the boundary of $\Omega_a$. The Kraus
layer is \emph{numerically inactive} in this case, with
$P_{\mathrm{lost}}$ reduced to a numerical-diffusion floor that
vanishes with grid refinement. The Euler triad satisfies $T \equiv 0$
as well, but with $\gamma_j = 0$ for every $j$; the
bound~(\ref{eq:norm_bound}) degenerates to exact conservation,
$\norm{\vect{a}(t)}_2 = \norm{\vect{a}(0)}_2$, and trajectories orbit
on the intersection of the constant-norm sphere with the
constant-energy ellipsoid fixed by the initial condition. Conservation
of $\norm{\vect{a}}_2$ alone does not bound the individual mode
amplitudes. Energy redistribution among the three modes can drive an
individual $|a_j(t)|$ above its initial value $|\bar{a}_j|$, with the
maximum set by the joint $(Z, E)$ conservation. The wavefunction
support, centered on the trajectory mean with spread $\sigma$, can
therefore reach the absorbing layer at late times even though the
trajectory itself remains bounded, and the Kraus layer absorbs the
resulting tail. Its action remains numerically inactive in the same
sense as the dissipative case. $P_{\mathrm{lost}}$ is small relative
to the resolution error and vanishes with a larger domain.

In both regimes, the Kraus layer's contribution to the validation is
not active absorption but two structural properties whose value is
independent of how much probability flux the dynamics direct toward
the boundary. First, the CPTP completeness
relation~(\ref{eq:cptp}) provides exact probability accounting through
the residual $|P_{\mathrm{ret}} + P_{\mathrm{lost}} - 1|$, which is
maintained at the round-off floor of the simulation
(Sec.~\ref{sec:verification}); without this construction, probability
lost to numerical diffusion at the boundary would be silently
destroyed rather than accounted for. Second, the layer prevents the
periodic SBP operator from reintroducing PDF tails at the opposite
boundary as a spurious wrap-around artifact, an effect that becomes
visible at coarser resolutions or longer integration times even when
the modal norm itself is bounded (Sec.~\ref{sec:pdf_evolution}). The
demonstration of active absorption, where $P_{\mathrm{lost}}$ reflects
physical flux from a system with $T \not\equiv 0$, lies outside the
scope of the present validation and is identified as a target for
follow-up work.
Sec.~\ref{sec:verification} verifies the two regime assignments
numerically.

\section{Test cases}
\label{sec:models}

The three equations studied here arise from the same quadratic advection
nonlinearity $(\vect{u}\cdot\nabla)\vect{u}$, projected onto resonant
triads of Fourier modes satisfying
$\vect{k}_1 + \vect{k}_2 + \vect{k}_3 = \vect{0}$. The cyclic structure
of the resulting coupling coefficients ensures
$\sum_j C_j = 0$ (Appendices~\ref{app:ns} and~\ref{app:hm}), so all three
cases satisfy the structural condition $T \equiv 0$ of
Proposition~\ref{prop:norm_bound}. The cases differ in the constraint
imposed on the parent PDE (incompressibility for Navier--Stokes and
Euler, modified vorticity inversion for Hasegawa--Mima) and in whether
dissipation is present. This combination spans two physical mechanisms
identified in Sec.~\ref{sec:confinement}: dissipative trajectory
contraction (NS, HM) and Hamiltonian orbit confinement on the joint
$(Z, E)$-constant curve (Euler). In both cases the Kraus absorbing
layer is numerically inactive, with $P_{\mathrm{lost}}$ small relative
to the resolution error. Governing equations and their Galerkin
truncations are derived in Appendix~\ref{app:parent_pdes}; this section
presents the resulting modal systems and identifies their algebraic
properties.

\subsection{Navier--Stokes triad}
\label{sec:ns_model}
A three-mode truncation of the 2D vorticity equation
(Appendix~\ref{app:ns}) to wavevectors satisfying
$\vect{k}_1 + \vect{k}_2 + \vect{k}_3 = \vect{0}$ yields
\begin{equation}
\dot{a}_j = -\nu|\vect{k}_j|^2\, a_j + C_j\, a_m a_n,
\quad (j,m,n) \text{ cyclic},
\label{eq:ns_triad}
\end{equation}
with damping rate $\gamma_j = \nu|\vect{k}_j|^2$ and coupling coefficients
$C_j$ derived in Appendix~\ref{app:ns}, Eq.~(\ref{eq:ns_coupling}). The
telescoping identity~(\ref{eq:sum_cj_zero}) confirms $\sum_j C_j = 0$,
so Corollary~\ref{cor:triad_confinement} applies: every trajectory decays
monotonically in $\norm{\vect{a}}_2$ and the Kraus layer is numerically
inactive. The energy conservation condition $\sum_j C_j/|\vect{k}_j|^2 = 0$
(proved in Appendix~\ref{app:ns}, Eq.~(\ref{eq:sum_cj_energy_zero}))
further ensures that the nonlinear contribution to $\dot{E}$ vanishes,
where $E(\vect{a}) = \tfrac{1}{2}\sum_j a_j^2/|\vect{k}_j|^2$ is the
total modal kinetic energy of the triad. The nonlinearity redistributes
energy among the three modes without changing the total, so $E$ evolves
through viscous dissipation alone and every trajectory satisfies
$E(\vect{a}(t)) \leq E(\vect{a}(0))$ for all $t \geq 0$. The viscous term renders the phase-space flow compressible,
with $\nabla\cdot\vect{f} = -\sum_j\nu|\vect{k}_j|^2$. In the
Weyl-ordered generator, this compressibility arises through the commutator
structure of Eq.~(\ref{eq:kvn_discrete}) and requires no separate
discretization.
\subsection{Incompressible Euler triad}
\label{sec:euler_model}
Setting $\nu = 0$ in~(\ref{eq:ns_triad}) gives
\begin{equation}
\dot{a}_j = C_j a_m a_n, \quad (j,m,n) \text{ cyclic},
\label{eq:euler_triad}
\end{equation}
with the same interaction coefficients as the viscous case. The phase-space
flow is divergence-free ($\nabla\cdot\vect{f} = 0$), and both the
phase-space norm $Z = \tfrac{1}{2}\norm{\vect{a}}_2^2$ and the modal
kinetic energy $E = \tfrac{1}{2}\sum_j a_j^2/|\vect{k}_j|^2$ are exactly
conserved by the nonlinear dynamics. Conservation of $Z$ follows from
$\dot{Z} = T(\vect{a}) = a_1 a_2 a_3\sum_j C_j = 0$, using
$\sum_j C_j = 0$. Conservation of $E$ follows from the separate
telescoping identity $\sum_j C_j/|\vect{k}_j|^2 = 0$
(Appendix~\ref{app:ns}, Eq.~(\ref{eq:sum_cj_energy_zero})). For the
test wavevectors used in Sec.~\ref{sec:results} the magnitudes
$|\vect{k}_j|$ are distinct, so $Z$ and $E$ are independent functions
on $\mathbb{R}^3$ and trajectories lie on the one-dimensional
intersection of the constant-$Z$ sphere with the constant-$E$
ellipsoid.
\subsection{Hasegawa--Mima triad}
\label{sec:hm_model}
The Hasegawa--Mima equation~\cite{Hasegawa1978} governs electrostatic
drift-wave turbulence in magnetized plasmas. Its triad truncation
(Appendix~\ref{app:hm}) to wavevectors satisfying
$\vect{k}_1 + \vect{k}_2 + \vect{k}_3 = \vect{0}$ gives
\begin{equation}
\dot{b}_j = -\mu\,|\vect{k}_j|^2\, b_j
+ C_j^{\mathrm{HM}}\, b_m b_n, \quad (j,m,n) \text{ cyclic},
\label{eq:hm_triad}
\end{equation}
where $b_j$ are potential vorticity amplitudes and $\mu$ is a collisional
damping rate. The coupling coefficients differ from the NS case by the
replacement $|\vect{k}|^{-2} \to (|\vect{k}|^2+1)^{-1}$ (Appendix~\ref{app:hm},
Eq.~(\ref{eq:hm_coupling})), reflecting the
modified potential-vorticity inversion introduced by the ion polarization
drift.
The difference structure of the HM coefficients is identical to that of
the NS coefficients, so the telescoping identity carries over to give
$\sum_j C_j^{\mathrm{HM}} = 0$ (Appendix~\ref{app:hm}).
Corollary~\ref{cor:triad_confinement} applies: every trajectory decays
monotonically in $\norm{\vect{b}}_2$ and the Kraus layer is numerically
inactive. The phase-space compressibility is
$\nabla\cdot\vect{f} = -\mu\sum_j|\vect{k}_j|^2$, structurally
identical to the NS case.

\section{Numerical results}
\label{sec:results}

Each of the KvN solutions was compared to a Monte Carlo (MC)
simulation of the respective parent ODE~(\ref{eq:ode}). A population of $N_s$
independent trajectories $\vect{a}^{(s)}(t)$ ($s = 1,\ldots,N_s$) was
integrated from samples of the initial distribution
$\rho(\vect{a},0)$ using a classical time integrator (fourth-order
Runge--Kutta, RK4, with the same time step as the KvN solver). Moments of $\rho(\vect{a},t)$ were estimated by ensemble averaging,
$\langle g \rangle_{\mathrm{MC}}(t) \equiv N_s^{-1}\sum_s g(\vect{a}^{(s)}(t))$,
with statistical error $\mathrm{SE} = O(N_s^{-1/2})$. The MC estimator
converges to the exact Liouville evolution in the limit $N_s \to \infty$,
$\Delta t \to 0$. The initial ensemble samples $\rho(\vect{a},0)$ such that each
trajectory is an exact characteristic of Eq.~(\ref{eq:liouville}), and
the empirical density of the ensemble is a consistent estimator of
$\rho(\vect{a},t)$. MC therefore serves as a model-free reference
against which the KvN discretization can be validated without assuming
any discretization-side quantity. The KvN discretization, in contrast,
solves the continuous PDE~(\ref{eq:schrodinger}) on a fixed
phase-space grid and reports the density $\rho = |\psi|^2$ directly;
comparison with MC tests the claim that the Hilbert-space-embedded
evolution reproduces the classical statistics to the accuracy allowed
by the grid and Krylov parameters. The $1/\sqrt{N_s}$ convergence of MC
is also the cost that KvN with amplitude estimation is designed to
circumvent for rare-event queries (Sec.~\ref{sec:quantum}), but for the
classical verification reported in this section MC is used at sample
counts large enough that its statistical error is well below the
KvN--MC discrepancies being reported.

\subsection{Computational setup}
\label{sec:setup}

The three triadic truncations span the dissipative and Hamiltonian
regimes identified in Sec.~\ref{sec:confinement}, and contrast fluid
(NS, Euler) with plasma (HM) coupling structure while remaining
simple enough for analytical verification against the predictions of
Secs.~\ref{sec:confinement} and~\ref{sec:models}.

The three triadic cases share wavevectors $\vect{k}_1 = (1,0)$, $\vect{k}_2 = (1,1)$,
$\vect{k}_3 = (-2,-1)$. Evaluating the coupling
formula~(\ref{eq:ns_coupling}) gives $C_1 = 0.3$, $C_2 = -0.8$,
$C_3 = 0.5$, used by the NS triad ($\nu = 0.1$) and the Euler triad
($\nu = 0$). For the HM triad, the same wavevectors evaluate
via~(\ref{eq:hm_coupling}) to $C_1^{\mathrm{HM}} = 1/6$,
$C_2^{\mathrm{HM}} = -1/3$, $C_3^{\mathrm{HM}} = 1/6$, with $\mu = 0.1$;
their smaller magnitudes compared to $|C_j| \leq 0.8$ for the NS case
reflect the reduced coupling strength due to the polarization drift. The
Kraus absorbing layer used a cosine-taper width $n_s = 3$.

All computations used double-precision arithmetic with fourth-order periodic
SBP operators (Appendix~\ref{app:sbp}) and a uniform time step
$\Delta t = 0.02$. The action $\mat{U}\psi$ was approximated by a Krylov
subspace projection~\cite{AlMohy2011} that builds a low-dimensional basis
from successive sparse matrix--vector products and computes the
matrix-exponential action within that basis without forming $\mat{U}$
explicitly. The Arnoldi residual was monitored at each step. For the NS
triad, the Krylov dimension required to reduce the residual below
$10^{-14}$ was $m \leq 12$, and was comparable for the other two cases.
Under this tolerance, the error in $\norm{\mat{U}\psi}$ was bounded by a
quantity of the same order~\cite{AlMohy2011}, so the observed unitarity
residual at the $10^{-15}$ level reflects floating-point rounding rather
than Krylov truncation.

The initial wavefunction was $\psi(\vect{a},0) =
\sqrt{\rho_0(\vect{a})}$, where $\rho_0(\vect{a}) \propto
\exp\bigl(-\tfrac{1}{2}\sum_j (a_j - \bar{a}_j)^2/\sigma^2\bigr)$ is a
Gaussian probability density with standard deviation $\sigma = 0.15$ in
each coordinate, centered at a specified mean $\bar{\vect{a}}$. For the
triad cases, the phase-space domain was defined by the cube $[-h, h]^3$ centered at
the origin with half-width $h = a_{0,\max} + 6\sigma$, where
$a_{0,\max} = \max_j\bar{a}_j$ is the largest initial modal amplitude.
For $\bar{\vect{a}} = (1.0, 0.8, 0.5)$ this gives $h = 1.9$ and a domain
$[-1.9, 1.9]^3$. The minimum margin from any mode's initial mean to the
nearest boundary is exactly $6\sigma$. The $6\sigma$ margin
places the sponge boundary at $P_{\mathrm{tail}} =
\tfrac{1}{2}\mathrm{erfc}(6/\sqrt{2}) \approx 10^{-9}$ initial
probability from the nearest boundary, making boundary truncation
negligible relative to all other error sources.
Statistical moments were obtained from $\avg{g} = \sum_{\vect{i}}
g(\vect{a}_{\vect{i}}) |\psi_{\vect{i}}|^2 \Delta V$. Monte Carlo
references were computed from RK4 trajectories sampled from the same
Gaussian initial condition.

All three triad cases ($N=3$) used $M = 40$ grid points per dimension,
giving $K = 64{,}000$ and $\sigma/\Delta a = 0.15 \times 40/3.8 = 1.58$
uniformly across the NS, Euler, and HM systems.

Monte Carlo sample sizes were chosen so that the standard error
$\mathrm{SE}(E) = \sigma_E/\!\sqrt{N_s}$ lay below the KvN--MC validation
discrepancy at the validation grid~\cite{Fishman1996}. The
per-trajectory energy standard deviation $\sigma_E \approx 0.163$ for
the triads was estimated from a pilot run of $10^4$ NS-triad
trajectories at $M = 40$. With $N_s = 100{,}000$ this gives
$\mathrm{SE}(E) \approx 5.2 \times 10^{-4}$ ($0.07\%$ of $\avg{E}$),
a factor of three below the ${\sim}0.2\%$ KvN--MC offset observed at
$M = 40$ (Sec.~\ref{sec:ns_results}). This yielded a signal-to-noise
ratio of ${\approx} 2.8$. Grid convergence was validated against a
dedicated high-sample reference of $N_s = 10^6$ trajectories.

Table~\ref{tab:summary} summarizes the case-specific parameters and
principal diagnostics.

\begin{table}[tbp]
\caption{Computational parameters and diagnostics: number of modes, $N$,
   grid points per dimension, $M$,  total grid points, $K = M^N$,
  simulation horizon, $T$,  MC sample count, $N_s$,
  cumulative probability absorbed by the Kraus layer, $P_{\mathrm{lost}}(T)$, and qubit count,
  $n_q = N\lceil\log_2 M\rceil$. All cases use
  $\sigma = 0.15$, a domain margin of $n=6$ standard deviations
  (Sec.~\ref{sec:domain_sizing}), and fourth-order SBP.}
\label{tab:summary}
\begin{ruledtabular}
\begin{tabular}{lccccccc}
System & $N$ & $M$ & $K$ & $T$ & $N_s$ & $P_{\mathrm{lost}}(T)$ & $n_q$ \\
\hline
NS triad    & 3 & 40 & 64\,000  & 0.5 & 100\,000 & $3.0 \times 10^{-6}$  & 18 \\
Euler triad & 3 & 40 & 64\,000  & 1.0 & 100\,000 & $1.7 \times 10^{-5}$  & 18 \\
HM triad    & 3 & 40 & 64\,000  & 0.5 & 100\,000 & $2.6 \times 10^{-6}$  & 18 \\
\end{tabular}
\end{ruledtabular}
\end{table}

\subsection{Verification studies}
\label{sec:verification}

\subsubsection{Phase-space norm bound}

Proposition~\ref{prop:norm_bound} predicts two distinct confinement
regimes for the three test cases. Figure~\ref{fig:phase_space_trajectories}
shows a time history of the worst-case norm ratio
$\max_{\mathrm{corners}}\norm{\vect{a}(t)}_2/\norm{\vect{a}(0)}_2$ over
all $2^N$ corners of the support box $\bar{\vect{a}} \pm 6\sigma$, where
$\bar{a}_j$ is the initial condition mean in coordinate $j$, the most
stringent test of confinement.

For the NS and HM triads, the ratio decays monotonically below unity
throughout, confirming Corollary~\ref{cor:triad_confinement}. The two
envelopes are nearly indistinguishable, confirming that confinement is
governed by the shared algebraic condition $\sum_j C_j = 0$ rather than
system-specific physics. For the Euler triad, the ratio holds at exactly
unity, consistent with exact conservation of $\norm{\vect{a}}_2^2$.

\begin{figure}[tbp]
  \includegraphics[width=\columnwidth]{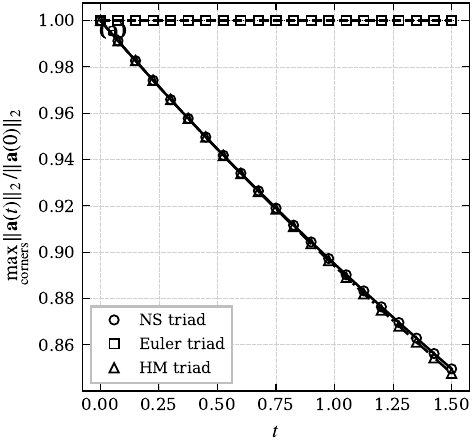}
  \caption{Worst-case norm ratio
    $\max_{\mathrm{corners}}\norm{\vect{a}(t)}_2/\norm{\vect{a}(0)}_2$
    over all $2^N$ corners of the support box $\bar{\vect{a}} \pm 6\sigma$:
    NS triad ($\circ$), Euler triad ($\square$),
    HM triad ($\triangle$). NS and HM envelopes decay monotonically
    (Corollary~\ref{cor:triad_confinement}); the Euler envelope holds at
    unity, consistent with exact conservation of $\norm{\vect{a}}_2^2$
    in the inviscid limit.}
\label{fig:phase_space_trajectories}
\end{figure}

\subsubsection{Grid convergence}
\label{sec:convergence}

The NS triad was considered for convergence analysis for three
reasons: it has a compressible phase-space flow ($\nabla\cdot\vect{f}
\neq 0$) unlike the Euler triad, it exercises the full
three-dimensional tensor-product grid, and its
antisymmetric coupling ensures $P_{\mathrm{lost}}$ is negligible at
the $6\sigma$ margin, isolating resolution error as the sole source of
discrepancy. Figure~\ref{fig:convergence} reports grid convergence for
the NS triad at $t = 0.5$ using second- and fourth-order periodic SBP
operators on grids $M = 16, 20, 24, 28, 32, 40$ ($K = M^3$ from 4\,096
to 64\,000). The error in expected modal kinetic energy
$\avg{E} = \avg{\tfrac{1}{2}\sum_j a_j^2/|\vect{k}_j|^2}$ decreases as
$(\Delta a)^p$ with pairwise rates $p = 2.00$ (second-order) and
$p = 4.00$ (fourth-order) at the finest grid pair ($M = 32 \to 40$),
consistent with the truncation error formulas
Eqs.~(\ref{eq:sbp2_truncation}) and~(\ref{eq:sbp4_truncation}). The
convergence rate is set by the SBP stencil, not by the KvN formulation.
At the coarsest grids ($\sigma/\Delta a = 0.63$, $M = 16$) the
fourth-order rates are sub-asymptotic because the initial
Gaussian occupies less than one grid cell per standard deviation; design
order $p = 4.00$ is reached once $\sigma/\Delta a \gtrsim 1$
($M \geq 28$). Second-order rates follow the same pattern
($p = 2.00$ at $M = 32 \to 40$).

Two independent references confirm convergence to the correct continuum
limit. Richardson extrapolation (RE) from the two finest grids ($M = 32, 40$)
gives $E^{\mathrm{RE}} = 0.6151$; a high-sample MC computation ($10^6$
trajectories) gives $E^{\mathrm{MC}} = 0.6149$. The two references agree
to $0.02\%$, confirming that both operator orders converge to the same
continuum solution. At $M = 40$, the fourth-order energy error against
$E^{\mathrm{RE}}$ is $4.1 \times 10^{-4}$, a factor of $10.7$ below the
$4.4 \times 10^{-3}$ second-order value, reflecting the combined effect
of the higher convergence order and the smaller leading constant of the
fourth-order stencil.

The probability loss [Fig.~\ref{fig:convergence}(b)] decreases
monotonically with refinement, from
$P_{\mathrm{lost}} = 1.2 \times 10^{-2}$ at $M = 16$ to
$3.0 \times 10^{-6}$ at $M = 40$ for the fourth-order operator. These
values reflect numerical diffusion from the finite-order discretization
spreading the PDF tails: by Proposition~\ref{prop:norm_bound}, physical
NS-triad trajectories cannot reach the boundary for $n = 6$, so the
absorbed probability is entirely a numerical artifact that vanishes with
refinement. The conservation residual
$|P_{\mathrm{ret}} + P_{\mathrm{lost}} - 1| < 10^{-15}$
at every $(M, \mathrm{order})$ combination, consistent with the Krylov
tolerance and floating-point rounding
(Theorem~\ref{thm:discrete_unitarity}).

\begin{figure*}[tbp]
  \includegraphics[width=\textwidth]{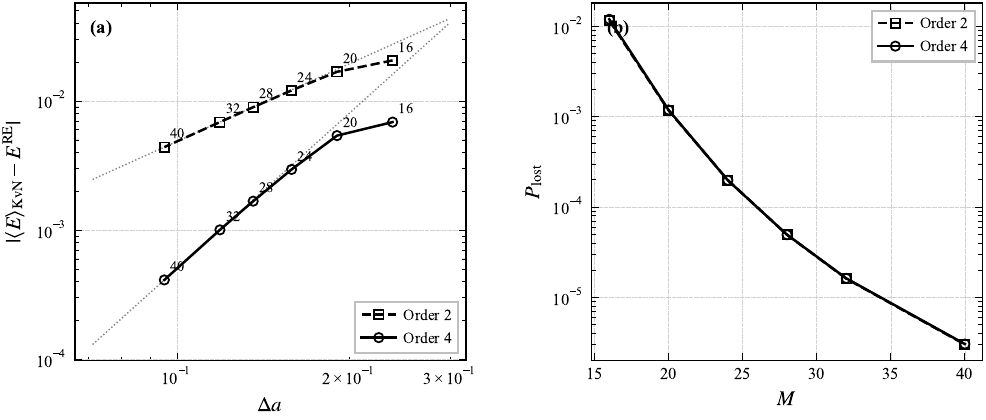}
  \caption{Grid convergence for the NS triad ($\nu = 0.1$, $t = 0.5$,
    $n=6$ domain, $K = M^3$). Numbers at each point indicate $M$.
    (a) Error in $\avg{E} = \avg{\tfrac{1}{2}\sum_j a_j^2/|\vect{k}_j|^2}$
    against Richardson extrapolation from the two finest grids ($M=32,40$);
    fourth-order (solid) and second-order (dashed) SBP operators.
    (b) Cumulative $P_{\mathrm{lost}}$. The conservation residual
    $|P_{\mathrm{ret}} + P_{\mathrm{lost}} - 1| < 10^{-15}$ at every step.}
  \label{fig:convergence}
\end{figure*}

\subsubsection{Domain sizing}
\label{sec:domain_sizing}

Choosing the domain extent $\Omega_a$ involves two competing effects.
Enlarging the box reduces the probability truncated at $t = 0$ but
coarsens the grid spacing $\Delta a$ at fixed $M$, degrading the
resolution of the PDF. In practice the dissipative dynamics contract the
PDF rapidly, leaving resolution error to dominate at any adequately-sized
domain.

For the triad cases, the domain was the cube $[-h, h]^3$ centered at the
origin, with half-width $h = a_{0,\max} + n\sigma$, where $n$ is the
margin parameter and $a_{0,\max} = 1.0$ is the largest component of the
initial mean. This zero-centered construction is natural for Galerkin
modal amplitudes, which can take both positive and negative values. The
mode with the largest initial amplitude ($a_{0,\max} = 1.0$) lies exactly
$n\sigma$ from the right boundary; all other modes have larger margins
($\geq 7.3\sigma$ for $n=6$). The phase-space resolution was
$\sigma/\Delta a = M\sigma/(2h)$.

For the standard case $n = 6$: $h = 1.9$, $\Delta a = 3.8/M$, and
$\sigma/\Delta a = 0.15M/3.8$. For the triad validation cases at $M = 40$,
this gives $\sigma/\Delta a = 1.58$; at $M = 16$ (the coarsest grid in the
convergence study) it gives $\sigma/\Delta a = 0.63$, below the threshold
for resolved initial conditions.

The probability outside the domain is dominated by the single closest
boundary, the right side of mode~1, which lies $6\sigma$ from the initial
mean $\bar{a}_1 = 1.0$. Using the complementary error function,
$P(a_1 > 1.9) = \tfrac{1}{2}\mathrm{erfc}(6/\sqrt{2}) \approx 10^{-9}$.
Contributions from the remaining boundaries are smaller by several orders
of magnitude. The total initial tail probability is therefore
\begin{equation}
P_{\mathrm{tail}}(6) \approx \tfrac{1}{2}\mathrm{erfc}(6/\sqrt{2})
\approx 10^{-9},
\label{eq:ptail}
\end{equation}
where $\mathrm{erf}(x) = (2/\sqrt{\pi})\int_0^x e^{-s^2}\,\dd s$. Since
the dissipative dynamics contract the PDF over time, this sets a strict
upper bound on the probability that can ever reach the boundary.

Figure~\ref{fig:domain_study} reports the domain study for the NS triad
($M = 40$, $t = 0.5$, 200\,000 MC trajectories) over margins from $n=2$
to $n=8$ at fixed grid count. The energy error is dominated by two
distinct mechanisms: at small $n$ the absorbing layer encroaches on the
initial condition and contaminates the retained probability, while at
large $n$ the resolution $\sigma/\Delta a$ falls and finite-stencil
truncation error grows. The result is a non-monotone trend with a
minimum near $n = 5$--$6$, beyond which resolution controls the error
budget. At $n=2$ the NS triad yielded $|\Delta E|/E^{\mathrm{MC}} =
7.8\%$ ($\sigma/\Delta a = 2.31$, sponge inner edge at $0.70\sigma$ from
$\bar{a}_1$); at the minimum ($n = 6$) the error was $0.012\%$
($\sigma/\Delta a = 1.58$); at $n = 8$ it had risen back to $0.07\%$
($\sigma/\Delta a = 1.36$).

The probability absorbed by the Kraus layer fell monotonically with $n$
across the range studied at $M = 40$. At the smallest margin $n = 2$
($h = 1.3$, $\Delta a = 0.065$), the sponge inner edge lay at $1.105$,
only $0.70\sigma$ from $\bar{a}_1 = 1.0$, so the sponge attenuated the
initial Gaussian directly and absorption reached $P_{\mathrm{lost}} =
15\%$. At $n = 3$ the sponge retreated to $1.55\sigma$ from mode~1 and
$P_{\mathrm{lost}}$ fell to $3.2\%$. Further increases gave
$P_{\mathrm{lost}} = 1.5 \times 10^{-4}$ at $n = 5$,
$3.0 \times 10^{-6}$ at $n = 6$,
$2.9 \times 10^{-8}$ at $n = 7$, and
$1.4 \times 10^{-10}$ at $n = 8$. The absence of a minimum within the
range reflects that at $M = 40$ the resolution remained sufficient
($\sigma/\Delta a \geq 1.36$ even at $n = 8$) to keep numerical
diffusion at the boundary below the one-sided Gaussian tail at every
margin.

The practical recommendation is to select $n$ such that
$\sigma/\Delta a \gtrsim 1$ while keeping the sponge well clear of the
initial condition. For the NS triad at $M = 40$, the energy error fell
below $0.5\%$ for $n \geq 4$ and reached its minimum near $n = 5$--$6$;
the absorbed probability fell below $10^{-5}$ for $n \geq 6$. The
$n = 6$ margin used for all validation cases placed the sponge inner
edge at $4.1\sigma$ from mode~1 and gave $\sigma/\Delta a = 1.58$,
chosen to provide clean convergence measurements without boundary
interaction. Smaller margins of $n = 4$ to $5$ give a finer grid
($\sigma/\Delta a \approx 1.7$ to $1.9$ at $M = 40$) at the cost of
larger boundary absorption (Fig.~\ref{fig:domain_study}b). For the
quantum implementation, smaller domains permit finer resolution at
fixed qubit count, since the qubit cost is determined by $M$ through
$n_q = N\lceil\log_2 M\rceil$.

\begin{figure*}[tbp]
  \includegraphics[width=\textwidth]{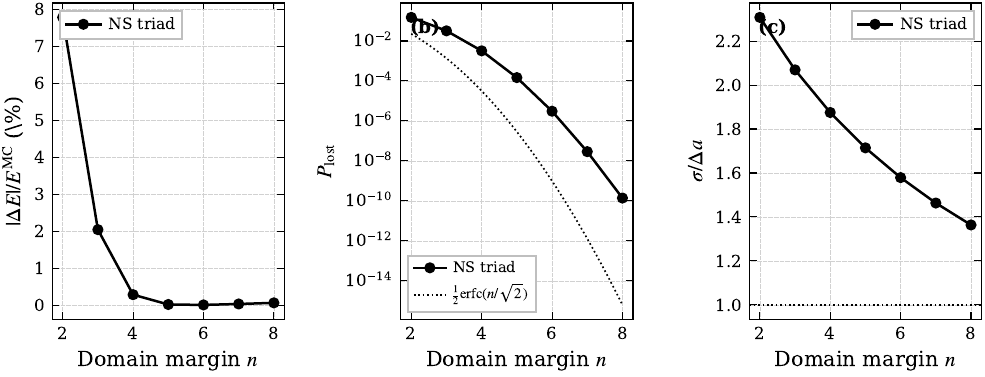}
  \caption{Domain sizing study at fixed grid count for the NS triad
    ($M=40$, $\nu=0.1$, $t=0.5$). The horizontal axis is the margin
    parameter $n$ in $h = a_{0,\max} + n\sigma$.
    (a) Energy error relative to MC: large at small $n$ (sponge contacts
    the initial condition) and at large $n$ (resolution loss), with a
    minimum near $n = 5$--$6$.
    (b) $P_{\mathrm{lost}}$ falls monotonically with $n$. The
    one-sided Gaussian tail~(\ref{eq:ptail}) (dotted) underestimates
    the absorption throughout the range studied, with the gap
    dominated by sponge contact at small $n$ and by residual
    numerical diffusion at large $n$.
    (c) $\sigma/\Delta a$ decreases with domain extent; the dashed
    line marks $\sigma/\Delta a = 1$.}
\label{fig:domain_study}
\end{figure*}

\subsection{Navier--Stokes triad}
\label{sec:ns_results}

The NS triad (Sec.~\ref{sec:ns_model}) was computed at $M = 40$
($K = 64{,}000$) with $\nu = 0.1$ and initial condition
$\bar{\vect{a}} = (1.0, 0.8, 0.5)$.

Figure~\ref{fig:ns_validation} shows KvN with MC ($N_s = 100{,}000$
trajectories) over $t \in [0, 0.5]$. Mean amplitudes agreed to within
$0.1\%$; the modal kinetic energy $\avg{E}$ and the phase-space norm
$\avg{Z}$ each tracked the MC reference to within $0.2\%$,
consistent with the $\sigma/\Delta a = 1.58$ resolution at $M = 40$.
Probability accounting was confirmed in Fig.~\ref{fig:ns_validation}(d):
the conservation residual $|P_{\mathrm{ret}} + P_{\mathrm{lost}} - 1| <
4 \times 10^{-16}$ at every step. The probability loss was
$P_{\mathrm{lost}} = 3.0 \times 10^{-6}$ at $t = 0.5$, consistent with
the numerically inactive Kraus layer of Corollary~\ref{cor:triad_confinement}.
The loss was purely numerical diffusion at $\sigma/\Delta a = 1.58$; no
physical trajectory reaches the boundary for $n = 6$.

\begin{figure*}[tbp]
  \includegraphics[width=0.49\textwidth]{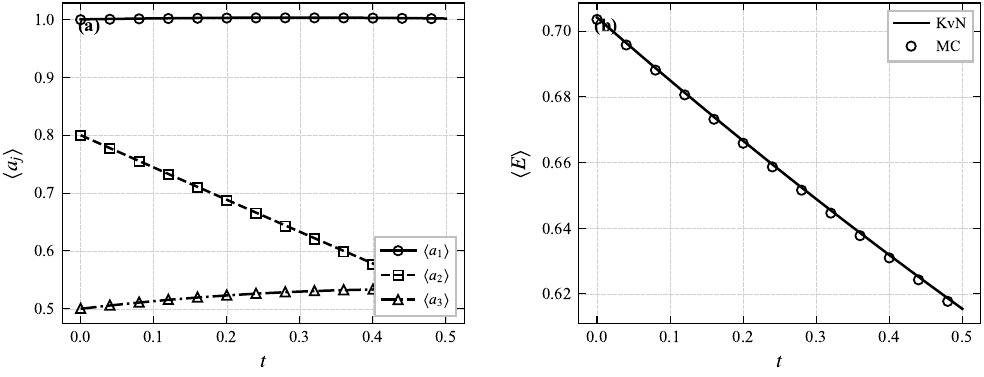}
  \includegraphics[width=0.49\textwidth]{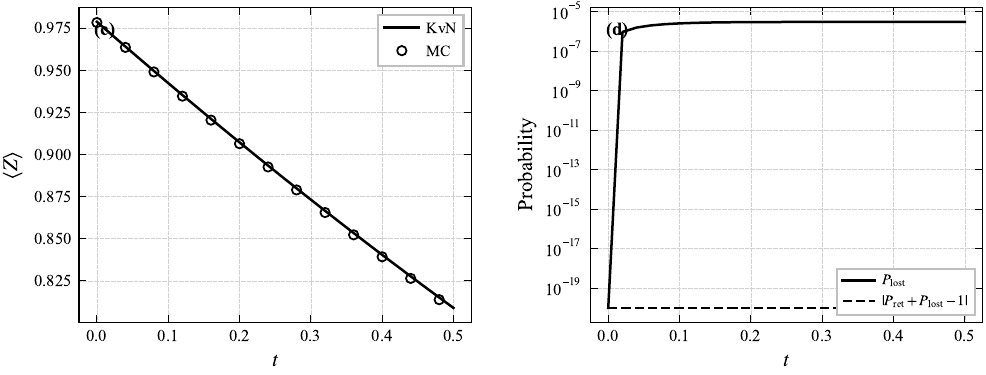}
  \caption{NS triad ($M=40$, $\nu=0.1$, fourth-order SBP).
    (a) Mean amplitudes.
    (b) Modal kinetic energy
    $\avg{E} = \avg{\tfrac{1}{2}\sum_j a_j^2/|\vect{k}_j|^2}$.
    (c) Squared amplitude norm
    $\avg{Z} = \avg{\tfrac{1}{2}\norm{\vect{a}}_2^2}$.
    (d) $P_{\mathrm{lost}}$ (solid) and conservation residual
    $|P_{\mathrm{ret}} + P_{\mathrm{lost}} - 1|$ (dashed).
    Lines: KvN; symbols: MC ($N_s = 100{,}000$ trajectories).}
\label{fig:ns_validation}
\end{figure*}

To show the qualitative PDF behavior beyond the benchmarking
window\label{sec:pdf_evolution}, the same triad was integrated to
$t = 1.5$ at a coarser resolution $M = 28$, chosen to keep the full
three-dimensional grid tractable for visualization over the extended
horizon. Figure~\ref{fig:pdf_evolution} shows the marginal density in
the $(a_1, a_2)$ plane, integrated over $a_3$, with $\Delta t = 0.02$.
The initial Gaussian developed multiple peaks by $t \approx 0.9$,
producing features finer than the grid spacing at late times. This
sub-grid structure is the dominant source of error identified in
Sec.~\ref{sec:convergence}. The finite-order SBP approximation smooths
features below $\Delta a$, introducing dispersion errors in the moments
that diminish with grid refinement.

The probability loss at $M = 28$ increased from $4.9 \times 10^{-5}$ at
$t = 0.5$ to $6.5 \times 10^{-5}$ at $t = 1.5$ (compare $3.0 \times
10^{-6}$ at $M = 40$ in Table~\ref{tab:summary}). At late times, nonlinear
mode coupling created additional probability peaks that extended into the
domain margin. The encroachment occurred not because the norm bound was violated, but
because the PDF developed structure at individual component values near the
box boundary while $\norm{\vect{a}}_2$ remained bounded. Without the Kraus
layer, the periodic SBP operator would reintroduce this probability at the
opposite boundary as a spurious wrap-around artifact.

\begin{figure}[tbp]
  \includegraphics[width=\columnwidth]{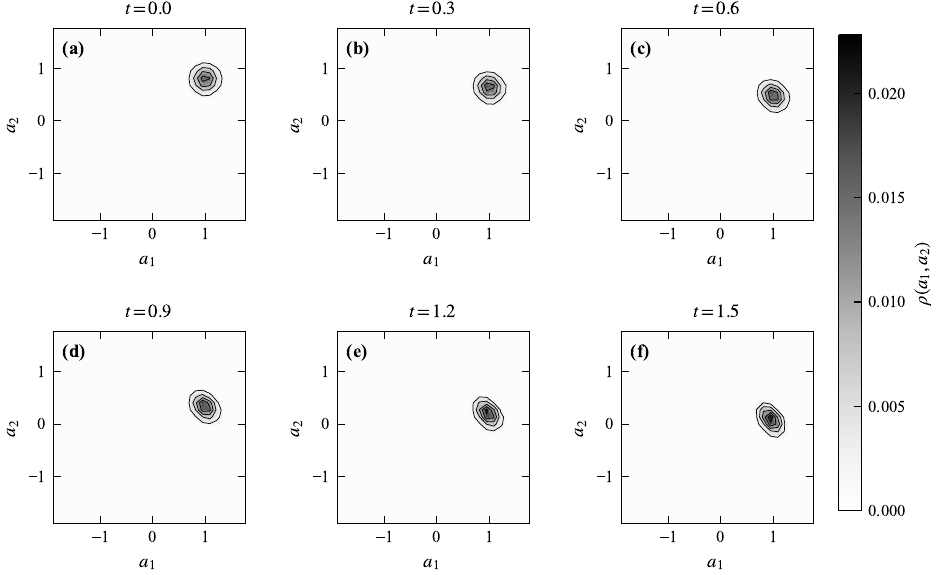}
  \caption{PDF evolution for the NS triad ($M=28$, $\nu=0.1$, fourth-order
    SBP): marginal density in the $(a_1, a_2)$ plane, integrated over
    $a_3$. The initial Gaussian develops multiple peaks through nonlinear
    energy transfer. $P_{\mathrm{lost}} = 6.5 \times 10^{-5}$ at
    $t = 1.5$.}
\label{fig:pdf_evolution}
\end{figure}

\subsection{Incompressible Euler triad}
\label{sec:euler_results}

The Euler triad (Sec.~\ref{sec:euler_model}) was computed at $M = 40$
($K = 64{,}000$) with initial condition $\bar{\vect{a}} = (1.0, 0.8, 0.5)$
and $6\sigma$ margins, over $t \in [0, 1.0]$, with the Kraus absorbing
layer applied as in the dissipative cases.

Figure~\ref{fig:euler_triad} shows the Euler triad dynamics with no
viscous decay. The mean amplitudes redistributed energy among modes.
$\avg{a_1}$ grew from $1.0$ to $1.10$,
$\avg{a_2}$ fell from $0.8$ to $0.27$, and $\avg{a_3}$ rose from $0.5$
to $0.78$, with $\sum_j a_j^2$ preserved per trajectory by the
inviscid dynamics. Both KvN and MC maintained flat $\avg{E}$ and
$\avg{Z}$ on the $\pm 0.7\%$ scale shown in panels~(b)--(c). The KvN
energy decreased by $2.0 \times 10^{-4}$ ($0.028\%$) over
$t \in [0, 1]$, consistent with the $O(\Delta a^4)$ spatial truncation
error at $\sigma/\Delta a = 1.58$ and indistinguishable from MC
sampling scatter ($\mathrm{SE}(E) \approx 5.2 \times 10^{-4}$). The
probability loss was $P_{\mathrm{lost}} = 1.7 \times 10^{-5}$ at
$t = 1.0$, slightly larger than for the dissipative triads because the
absence of damping permits more extensive mode redistribution toward
the boundary; the absorbed mass nonetheless remained negligible
relative to the resolution error. The conservation residual
$|P_{\mathrm{ret}} + P_{\mathrm{lost}} - 1| < 2.3 \times 10^{-16}$
throughout, confirming exact probability accounting via the CPTP
completeness relation.

\begin{figure*}[tbp]
  \includegraphics[width=0.49\textwidth]{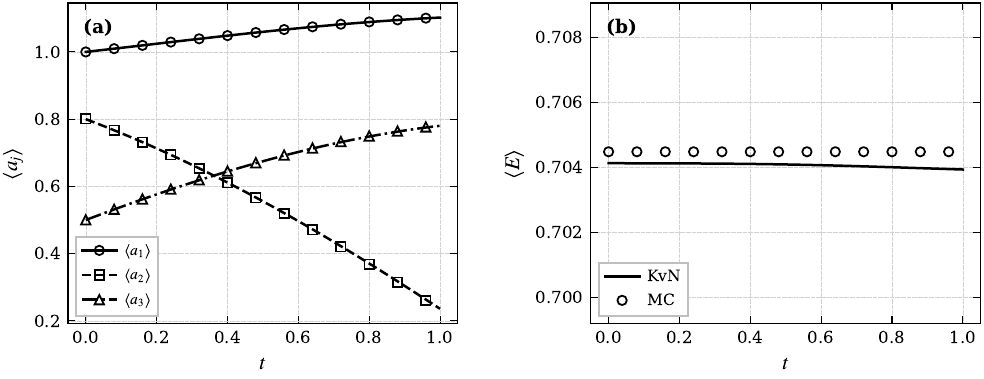}
  \includegraphics[width=0.49\textwidth]{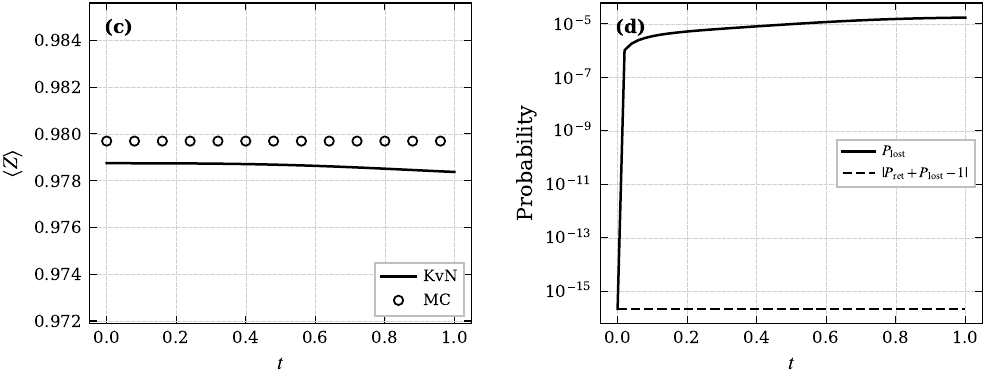}
  \caption{Euler triad ($M=40$, $\nu=0$, fourth-order SBP, Kraus layer
    applied).
    (a) Mean amplitudes. (b) Modal kinetic energy on a $\pm 0.7\%$ scale
    about the initial expected value $E_0 \equiv \avg{E}|_{t=0} = 0.7041$.
    (c) Squared amplitude norm
    $Z = \tfrac{1}{2}\norm{\vect{a}}_2^2$ on a $\pm 0.7\%$ scale about
    $Z_0 \equiv \avg{Z}|_{t=0} = 0.9788$.
    (d) $P_{\mathrm{lost}}$ (solid) and conservation residual
    $|P_{\mathrm{ret}} + P_{\mathrm{lost}} - 1|$ (dashed).
    The $0.028\%$ KvN energy drift over $t \in [0,1]$ is consistent
    with $O(\Delta a^4)$ spatial truncation and lies within the MC
    standard error.
    Lines: KvN; symbols: MC ($N_s = 100{,}000$ trajectories).}
\label{fig:euler_triad}
\end{figure*}

\subsection{Hasegawa--Mima triad}
\label{sec:hm_results}

The HM triad (Sec.~\ref{sec:hm_model}) was computed at the same resolution
and domain margins as the NS triad ($M = 40$, $K = 64{,}000$, $\mu = 0.1$,
$\bar{\vect{b}} = (1.0, 0.8, 0.5)$).

Figure~\ref{fig:hm_triad} compares KvN with MC ($N_s = 100{,}000$
trajectories) over $t \in [0, 0.5]$. The weaker HM coupling
($|C_j^{\mathrm{HM}}| \leq 1/3$ versus $|C_j| \leq 0.8$) produced slower
energy transfer among modes. Agreement with MC and conservation residuals
matched the NS case (Table~\ref{tab:summary}). The Kraus layer was numerically
inactive, with $P_{\mathrm{lost}} = 2.6 \times 10^{-6}$ at $t = 0.5$, slightly
below the NS value ($3.0 \times 10^{-6}$), consistent with the weaker
nonlinear coupling producing slower PDF spreading toward the boundary.

\begin{figure*}[tbp]
  \includegraphics[width=0.49\textwidth]{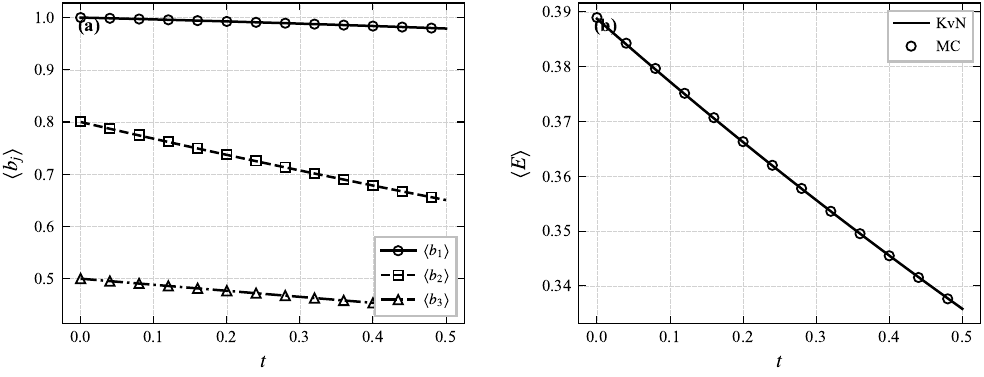}
  \includegraphics[width=0.49\textwidth]{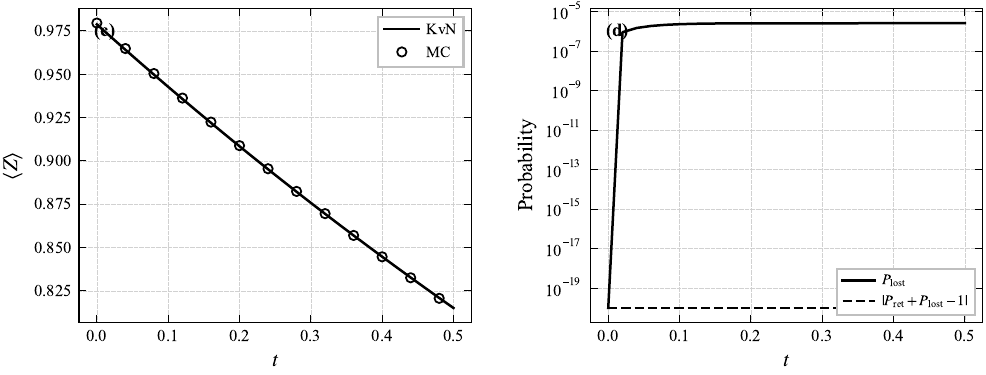}
  \caption{Hasegawa--Mima triad ($M=40$, $\mu=0.1$, fourth-order SBP).
    (a) Mean amplitudes.
    (b) Energy $\avg{E} = \avg{\tfrac{1}{2}\sum_j b_j^2/(|\vect{k}_j|^2+1)}$.
    (c) Squared amplitude norm
    $\avg{Z} = \avg{\tfrac{1}{2}\norm{\vect{b}}_2^2}$.
    (d) $P_{\mathrm{lost}}$ and conservation residual.
    Lines: KvN; symbols: MC ($N_s = 100{,}000$ trajectories).}
\label{fig:hm_triad}
\end{figure*}

\section{Implications for quantum implementation}
\label{sec:quantum}

This section sets out what the classical results of Secs.~\ref{sec:theory}
and~\ref{sec:results} imply for a quantum implementation of the KvN
evolution. No quantum algorithm was executed here. The validation studies
of Sec.~\ref{sec:results} were obtained on a classical CPU. The purpose was to
identify the algebraic and structural preconditions that would allow a
quantum implementation to deliver the asymptotic advantages for which KvN
is motivated, to estimate the resource counts that follow from the
specific discretization developed in this paper, and to position the
approach relative to alternative quantum strategies for classical
dynamics. A full gate-count analysis, circuit synthesis for the
state-dependent diagonal operators, and quantum-hardware execution are
left to future work.

The classical amplitudes $\psi_{\vect{i}}$ defined in
Sec.~\ref{sec:sbp} are amplitude-encoded on a quantum register of
$n_q = N\lceil\log_2 M\rceil$ qubits as
\begin{equation}
|\psi\rangle_S = \sum_{\vect{i}}\psi_{\vect{i}}|\vect{i}\rangle_S,
\label{eq:amplitude_encoding}
\end{equation}
with $|\vect{i}\rangle_S$ the computational-basis state corresponding to
the multi-index $\vect{i}$ of grid coordinates.

The algebraic structures established in
Secs.~\ref{sec:theory}--\ref{sec:results} serve as these preconditions:
exact anti-Hermiticity of the generator ensures that the propagator
remains unitary over the $O(\varepsilon^{-1})$ Grover iterates required
by amplitude estimation, and the CPTP boundary treatment provides exact
probability accounting with $O(1)$ ancilla overhead. The Kraus channel
of Sec.~\ref{sec:kraus} is non-unitary on the system Hilbert space
$\mathcal{H}_S = L^2(\Omega_a, \mathbb{C})$ and is realized on quantum
hardware through Stinespring dilation~\cite{Stinespring1955}. The
dilation embeds the channel as a unitary operation on the extended
Hilbert space $\mathcal{H}_S \otimes \mathcal{H}_E$, where
$\mathcal{H}_E$ is a single ancilla qubit (the environment register,
dimension 2) initialized in $|0\rangle_E$. The isometry
$\mat{V}: \mathcal{H}_S \to \mathcal{H}_S \otimes \mathcal{H}_E$
takes the explicit form
\begin{equation}
\mat{V} = \mat{K}_1 \otimes |0\rangle_E + \mat{K}_2 \otimes |1\rangle_E,
\label{eq:isometry_explicit}
\end{equation}
so that its action on any system state $|\psi\rangle_S$ gives
\begin{equation}
\mat{V}|\psi\rangle_S = \mat{K}_1|\psi\rangle_S \otimes |0\rangle_E
+ \mat{K}_2|\psi\rangle_S \otimes |1\rangle_E.
\label{eq:isometry}
\end{equation}
Measuring the ancilla in $|0\rangle_E$ yields the post-absorption state
$\mat{K}_1|\psi\rangle_S$, while outcome $|1\rangle_E$ registers the
absorbed fraction $\mat{K}_2|\psi\rangle_S$. The isometry property
$\mat{V}^\dagger \mat{V} = \mat{I}_S$ follows directly from the
completeness relation~(\ref{eq:cptp}):
\begin{align}
\mat{V}^\dagger \mat{V}
&= \bigl(\mat{K}_1^\dagger \otimes \langle 0|_E
       + \mat{K}_2^\dagger \otimes \langle 1|_E\bigr) \nonumber\\
&\quad\times \bigl(\mat{K}_1 \otimes |0\rangle_E
       + \mat{K}_2 \otimes |1\rangle_E\bigr) \nonumber\\
&= \mat{K}_1^\dagger\mat{K}_1\langle 0|0\rangle_E
 + \mat{K}_2^\dagger\mat{K}_2\langle 1|1\rangle_E \nonumber\\
&= \mat{K}_1^\dagger\mat{K}_1 + \mat{K}_2^\dagger\mat{K}_2
 = \mat{I}_S.
\label{eq:isometry_proof}
\end{align}
At each grid point $\vect{i}$, the ancilla undergoes a controlled $R_y$
rotation by angle $\theta_{\vect{i}} = \arccos(\sqrt{W_{\vect{i}}})$,
with post-selection on $|0\rangle_E$ retaining the system state and
$|1\rangle_E$ flagging absorption:
\begin{equation}
|i\rangle_S|0\rangle_E \;\longrightarrow\;
|i\rangle_S\bigl(\sqrt{W_{\vect{i}}}\,|0\rangle_E
+ \sqrt{1-W_{\vect{i}}}\,|1\rangle_E\bigr).
\label{eq:ancilla_rotation}
\end{equation}
Post-selection on $|0\rangle_E$ retains the evolved state with probability
$\norm{\mat{K}_1\psi}^2 \equiv 1 - \Delta P_{\mathrm{lost}}$, where
$\Delta P_{\mathrm{lost}}$ is the probability absorbed in a single step.
The total qubit count is $N\lceil\log_2 M\rceil + 1$, where the $+1$ accounts
for the single ancilla qubit. The gate complexity of loading the spatially
varying rotation angles $\theta_{\vect{i}}$ into the circuit depends on the
structure of $W_{\vect{i}}$ and on the available arithmetic synthesis
techniques. For the separable cosine taper~(\ref{eq:taper}), the weights
factorize across coordinate directions, which may reduce circuit depth
relative to a general $K$-point loading problem.

Realizing the propagator $\mat{U} = \exp(\mat{L}\Delta t)$ on quantum
hardware uses Hamiltonian simulation
algorithms~\cite{Berry2015,Low2017} applied to the Hermitian operator
$\mat{H} = -i\mat{L}$. The gate count is polylogarithmic in
$K = M^N$ at fixed sparsity but grows with the spectral norm
$\norm{\mat{L}}$, which scales as $\max|f_j|/\Delta a$ and increases
with grid refinement at fixed physical parameters. The state-dependent
diagonal operators $\mat{F}_j$ require controlled arithmetic circuits
whose cost scales polynomially with the polynomial degree of
$f_j(\vect{a})$. For the quadratic Navier--Stokes nonlinearity, the
circuit reduces to controlled multiply-and-add operations within
established synthesis techniques~\cite{Haner2018}. The generator
$\mat{L}$ has at most $4N$ nonzeros per row for the fourth-order
operator, giving 12 nonzeros per row for the three-mode triad
(Appendix~\ref{app:sbp_sparsity}). For $M = 40$ ($K = 64{,}000$), the
sparse $\mat{L}$ requires $12K \approx 768{,}000$ stored doubles
($\approx 6$\,MB), compared to $K^2 \approx 4.1 \times 10^9$ doubles
($\approx 33$\,GB) for a dense representation. The Krylov method
exploits this sparsity to advance the wavefunction in $O(K)$ time per
step with constant memory. The quantum advantage therefore becomes
asymptotically accessible as $M$ increases, where classical memory for
the dense propagator becomes prohibitive while the quantum state
requires only $N\lceil\log_2 M\rceil = 3 \times 6 = 18$ qubits for the
three-mode triads, with $\lceil\log_2 40\rceil = 6$ qubits per mode. For the
domains used here ($n = 6$, $M = 40$), the per-step absorption is
$\Delta P_{\mathrm{lost}} \sim 3.0 \times 10^{-6}/25 \approx
1.2 \times 10^{-7}$ for the NS triad, so the post-selection success
probability exceeds $1 - 1.2 \times 10^{-7}$ per step and is
effectively unity.

\subsection{Quantum advantage for distributional queries}
\label{sec:amplitude_estimation}

Once the evolved wavefunction $|\psi(t)\rangle_S$ is prepared on the
quantum register, distributional quantities can be extracted via quantum
amplitude estimation~\cite{Brassard2002}. These take the form
\begin{equation}
P(A) = \langle\psi(t)|\Pi_A|\psi(t)\rangle
= \sum_{\vect{i} \in A} |\psi_{\vect{i}}(t)|^2 \Delta V,
\label{eq:prob_query}
\end{equation}
where $\Pi_A = \sum_{\vect{i}\in A}|\vect{i}\rangle\langle\vect{i}|$ is
the projector onto the subset $A$ of the phase-space grid, representing
the probability that the system state lies in region $A$. Classical Monte
Carlo estimation of $P(A)$ requires $O(\varepsilon^{-2})$ samples to
achieve absolute accuracy $\varepsilon$. Quantum amplitude estimation
reduces this to $O(\varepsilon^{-1})$ queries to a controlled version of
the propagator, a quadratic speedup applicable to any probability or
expectation value computed from the evolved PDF.

The unitarity required for this speedup is the
$\mat{L} + \mat{L}^{\mathrm{T}} = 0$ identity of
Theorem~\ref{thm:discrete_unitarity}, which holds independently of grid
resolution and the divergence structure of the velocity field.
Classical validation at small $N$ therefore serves as a proof of
concept for quantum implementations at large $N$, where the exponential
compression of $M^N$ amplitudes into $N\lceil\log_2 M\rceil$ qubits is
the advantage.

The $O(\varepsilon^{-1})$ scaling is a Heisenberg-limited bound on the
precision achievable per unit of quantum resource. Heisenberg
uncertainty is absent from
grid design because the KvN wavefunction $\psi(\vect{a})$ is defined on
the $N$-dimensional modal amplitude space, on which the multiplication
operators $\hat{a}_j:\psi \mapsto a_j\psi$ commute and no conjugate
momenta are introduced; grid design is therefore governed by classical
resolution considerations, as confirmed by the convergence and
domain-sizing studies of Secs.~\ref{sec:convergence}
and~\ref{sec:domain_sizing}. The uncertainty principle re-enters at the
measurement stage and sets the fundamental limit on how much precision
can be extracted from coherent propagator applications.

The formulation is agnostic to the initial distribution
$\rho(\vect{a},0)$. Nothing in the KvN embedding, the Weyl-ordered
generator, the SBP discretization, or the Kraus layer is specific to
Gaussian initial conditions, and the quantum encoding
$\psi_{\vect{i}}(0) = \sqrt{\rho(\vect{a}_{\vect{i}},0)\,\Delta V}$
applies uniformly. State preparation costs depend on the distribution
class: uniform distributions prepare in $O(n_q)$ single-qubit gates;
efficiently integrable distributions (including Gaussians and many
physically motivated families) admit polynomial state preparation
through the Grover--Rudolph algorithm~\cite{Grover2002}; arbitrary
distributions require $O(2^{n_q})$ gates in the worst case and would
negate the exponential compression of the KvN encoding. The Gaussian
initial conditions used in Sec.~\ref{sec:results} belong to the
efficiently integrable class, so state preparation is not a bottleneck
for the applications considered here.

\subsection{Comparison with alternative quantum approaches}
\label{sec:comparison}

Both LCHS and Schr\"{o}dingerization apply to arbitrary non-unitary
linear PDE evolution. For the specific class
of Liouville problems studied here, the non-unitarity they address is
an artifact of the discretization rather than a feature of the
underlying physics. The difference lies in where the non-unitarity is
located and how it is handled.

In the LCHS framework, Novikau and Joseph~\cite{Novikau2025,Novikau2025b}
use upwind or centered finite differences for the KvN generator, producing
a non-Hermitian matrix that is decomposed as a linear combination of
Hamiltonian evolutions. This decomposition requires an auxiliary register
of $O(\log(1/\varepsilon))$ qubits and quantum signal processing circuits
whose depth scales with the spectral norm of the anti-Hermitian and
Hermitian parts separately. The approach is appropriate when the
discretization itself introduces non-unitarity.

\begin{sloppypar}
Schr\"{o}\-din\-ger\-iz\-a\-tion~\cite{Jin2023,Jin2024} lifts an arbitrary
$D$-dimensional linear PDE into a $(D\!+\!1)$-dimensional
Schr\"{o}dinger equation via a warped phase-space transformation. The
extra continuous dimension adds $O(\log(1/\varepsilon))$ qubits.
\end{sloppypar}

The present approach restricts to equations whose dynamics can be cast
as Liouville transport of a probability density. The Weyl-ordered
periodic SBP discretization is exactly anti-Hermitian
(Theorem~\ref{thm:discrete_unitarity}), so the interior propagator is
already unitary and neither the LCHS decomposition nor the
Schr\"{o}dingerization lift is required. The only non-unitary element
is the Kraus absorbing layer, which requires a single ancilla qubit;
the total qubit count is $N\lceil\log_2 M\rceil + 1$.

\section{Conclusions}
\label{sec:conclusions}

A discretization method for the Koopman--von Neumann formulation of
spectrally truncated fluid and plasma dynamics has been developed and
validated. The method yields exactly unitary operators through SBP
discretization of the Weyl-ordered generator.

The Weyl ordering was derived as the unique symmetric operator ordering
within the one-parameter family of weighted orderings that makes the
KvN generator anti-Hermitian in the complex Hilbert space over the
modal amplitude domain. Anti-Hermiticity requires real velocity fields,
a condition guaranteed by Galerkin projection of any real PDE onto real
basis functions.

Exact discrete unitarity was proved as a purely algebraic identity for
the Weyl-ordered periodic SBP generator
(Theorem~\ref{thm:discrete_unitarity}). The proof requires only that
the differentiation matrix be skew-symmetric and the velocity-field
matrix be real diagonal, and holds for any SBP stencil order, grid
resolution, velocity field, and divergence structure.

A split-step Kraus absorbing layer was constructed as a CPTP quantum
channel providing exact probability accounting at domain boundaries via
a Stinespring dilation requiring a single ancilla qubit. Two structural
properties of the layer carry independent value even when classical
dynamics confine trajectories away from the boundary: the CPTP
completeness relation enforces probability accounting at machine
precision, and the absorbing taper prevents the periodic SBP operator
from reintroducing PDF tails at the opposite boundary as a spurious
wrap-around artifact.

Validation on three equations spanning dissipative and Hamiltonian
regimes confirmed machine-precision interior unitarity and fourth-order
grid convergence. At the $n = 6$ domain margin, boundary absorption
was $\sim 3 \times 10^{-6}$ for the dissipative triads (NS, HM) and
$\sim 2 \times 10^{-5}$ for the Hamiltonian Euler triad, in all cases
negligible relative to the resolution error. The conservation residual
$|P_{\mathrm{ret}} + P_{\mathrm{lost}} - 1|$ was maintained below
$10^{-15}$ throughout all three cases, confirming exact probability
accounting via the CPTP completeness relation.

The guarantees of unitarity, probability conservation, and boundary
accounting follow from algebraic identities, not from numerical tuning.
The same approach extends to any Liouville system whose dynamics can be
cast in the Galerkin form studied here. Demonstration of an actively
absorbing Kraus layer on a system with $T \not\equiv 0$, such as a
forced truncation or a non-resonant mode set that breaks the cyclic
telescoping identity, is a natural extension of the present work and
is left to future investigation. The algebraic structures
established in this paper are the preconditions for a quantum
implementation of the KvN evolution; a complete gate-count analysis,
circuit synthesis for the state-dependent diagonal operators
$\mat{F}_j$, and quantum-hardware execution of the resulting propagator
are left to future work.


\appendix

\section{Periodic SBP first-derivative operators}
\label{app:sbp}

This appendix defines the two periodic SBP operators used throughout the
paper and establishes the properties they supply to the main results.
Appendix~\ref{app:sbp2} and~\ref{app:sbp_skewsym} give the
second-order operator; Appendix~\ref{app:sbp4} and~\ref{app:sbp_skewsym}
give the fourth-order operator.
Appendix~\ref{app:sbp_skewsym} verifies exact skew-symmetry
$\mat{D}^{(p)} + (\mat{D}^{(p)})^\mathrm{T} = 0$ for both, which is the
condition $\mat{D}_j^\dagger = -\mat{D}_j$ required by
Theorem~\ref{thm:discrete_unitarity}. The truncation errors stated in
Appendices~\ref{app:sbp2} and~\ref{app:sbp4} define the orders $p = 2$ and
$p = 4$ confirmed empirically in Sec.~\ref{sec:convergence}.
Appendix~\ref{app:sbp_sparsity} provides the nonzero counts cited in
Sec.~\ref{sec:quantum}.

On a uniform periodic grid with $M$ points and spacing $\Delta a$, a
diagonal-norm SBP first-derivative operator is defined by
\begin{equation}
\mat{D} = \mat{H}^{-1}\mat{Q}, \quad
\mat{Q} + \mat{Q}^{\mathrm{T}} = \mat{B},
\label{eq:sbp_def}
\end{equation}
where $\mat{H} = \Delta a\,\mat{I}$ is the norm matrix and $\mat{B}$
encodes boundary contributions~\cite{Kreiss1974,Strand1994,Svard2014}.
On periodic domains $\mat{B} = 0$, so the SBP property reduces to exact
skew-symmetry of $\mat{Q}$. The periodic SBP operator $\mat{D}$ is a real
circulant matrix of size $M \times M$, fully determined by its first row:
\begin{equation}
D_{ij} = w_{j-i \bmod M},
\label{eq:circulant_entry}
\end{equation}
where the weights $w_k$ are real. Exact skew-symmetry requires $w_{-k} =
-w_k$ for all $k$ (indices mod $M$); the explicit verification for the
specific weights defined in the subsections below is given in
Appendix~\ref{app:sbp_skewsym}.

\subsection{Second-order operator}
\label{app:sbp2}

The second-order periodic SBP operator has two nonzero weights per row:
\begin{equation}
w_{\pm 1} = \pm\frac{1}{2\Delta a}, \qquad
w_k = 0 \;\text{ for }\; \abs{k} \neq 1.
\label{eq:sbp2_weights}
\end{equation}
For $M = 6$:
\begin{equation}
\mat{D}^{(2)} = \frac{1}{2\Delta a}
\begin{pmatrix}
 0 &  1 &  0 &  0 &  0 & -1 \\
-1 &  0 &  1 &  0 &  0 &  0 \\
 0 & -1 &  0 &  1 &  0 &  0 \\
 0 &  0 & -1 &  0 &  1 &  0 \\
 0 &  0 &  0 & -1 &  0 &  1 \\
 1 &  0 &  0 &  0 & -1 &  0 \\
\end{pmatrix}.
\label{eq:sbp2_matrix}
\end{equation}
Skew-symmetry $\mat{D}^{(2)} + (\mat{D}^{(2)})^{\mathrm{T}} = 0$ is
manifest. The truncation error is, for any smooth function $g$:
\begin{equation}
(\mat{D}^{(2)}\vect{g})_i = g'(a_i)
+ \frac{(\Delta a)^2}{6}\,g'''(a_i) + O\!\left((\Delta a)^4\right).
\label{eq:sbp2_truncation}
\end{equation}

\subsection{Fourth-order operator}
\label{app:sbp4}

The fourth-order periodic SBP operator has four nonzero weights per row:
\begin{align}
w_{\pm 1} &= \pm\frac{2}{3\Delta a}, \qquad
w_{\pm 2} = \mp\frac{1}{12\Delta a}, \nonumber\\
w_k &= 0 \quad \text{for} \quad \abs{k} > 2.
\label{eq:sbp4_weights}
\end{align}
For $M = 6$:
\begin{equation}
\mat{D}^{(4)} = \frac{1}{12\Delta a}
\begin{pmatrix}
 0 &  8 & -1 &  0 &  1 & -8 \\
-8 &  0 &  8 & -1 &  0 &  1 \\
 1 & -8 &  0 &  8 & -1 &  0 \\
 0 &  1 & -8 &  0 &  8 & -1 \\
-1 &  0 &  1 & -8 &  0 &  8 \\
 8 & -1 &  0 &  1 & -8 &  0 \\
\end{pmatrix}.
\label{eq:sbp4_matrix}
\end{equation}
Skew-symmetry $\mat{D}^{(4)} + (\mat{D}^{(4)})^{\mathrm{T}} = 0$ is
again manifest. The truncation error is, for any smooth function $g$:
\begin{equation}
(\mat{D}^{(4)}\vect{g})_i = g'(a_i)
- \tfrac{(\Delta a)^4}{30}\,g^{(5)}(a_i)
+ O\!\left((\Delta a)^6\right).
\label{eq:sbp4_truncation}
\end{equation}

\subsection{Verification of exact skew-symmetry}
\label{app:sbp_skewsym}

This subsection verifies the condition $\mat{D}_j^\dagger = -\mat{D}_j$
required by Theorem~\ref{thm:discrete_unitarity} for the specific
operators defined in Appendices~\ref{app:sbp2} and~\ref{app:sbp4}.
The general circulant structure of Eq.~(\ref{eq:circulant_entry}) reduces
exact skew-symmetry to the antisymmetry of the weight sequence
$w_{-k} = -w_k$. Setting $\mat{Q}^{(p)} = \Delta a\,\mat{D}^{(p)}$ for
order $p \in \{2,4\}$ and reading off the weights:
\begin{itemize}
\item Second order: $w_1 = 1/(2\Delta a)$, $w_{-1} = -1/(2\Delta a) = -w_1$.
\item Fourth order: $w_1 = 2/(3\Delta a)$, $w_{-1} = -w_1$; $w_2 = -1/(12\Delta a)$, $w_{-2} = -w_2$.
\end{itemize}
The entry-level verification then gives
\begin{align}
Q^{(p)}_{ij} + Q^{(p)}_{ji}
= \Delta a\,(w_{j-i} + w_{-(j-i)}) = 0,
\label{eq:skewsym_verify}
\end{align}
for both operators, since $w_{-k} = -w_k$ in each case. The anti-Hermiticity
$\mat{D}^{(p)} + (\mat{D}^{(p)})^{\mathrm{T}} = 0$ then follows
immediately, providing the discrete counterpart
$\mat{D}_j^\dagger = -\mat{D}_j$ required by Eq.~(\ref{eq:D_antiherm})
and the proof of Theorem~\ref{thm:discrete_unitarity}.

\subsection{Sparsity and qubit cost}
\label{app:sbp_sparsity}

Both operators are sparse: $\mat{D}^{(2)}$ has 2 nonzeros per row and
$\mat{D}^{(4)}$ has 4 nonzeros per row. On the full $N$-dimensional
tensor-product grid of size $K = M^N$, the Weyl-ordered
generator~(\ref{eq:kvn_discrete}) inherits this sparsity via the
Kronecker structure~(\ref{eq:Dj_kronecker}). Each term
$\mat{F}_j\mat{D}_j + \mat{D}_j\mat{F}_j$ has the same sparsity pattern
as $\mat{D}_j$ itself, contributing $2s$ nonzeros per row, where
$s \in \{1,2\}$ is the stencil half-width. Since the Kronecker structures
of distinct coordinate directions are disjoint, summing over $N$ directions
gives at most $2sN$ nonzeros per row. For the three-mode triad ($N=3$)
with the fourth-order operator ($s=2$), this gives 12 nonzeros per row
out of $K = 64{,}000$ columns, confirming the sparsity that underpins the
polylogarithmic gate complexity of Hamiltonian simulation.

\section{Parent PDEs and Galerkin truncations}
\label{app:parent_pdes}

This appendix presents the parent partial differential equations and their
Galerkin truncations, making the structural origin of the coupling
coefficients explicit. The main result is the telescoping identity
$\sum_j C_j = 0$ for resonant triads derived from incompressible equations
(Appendices~\ref{app:ns} and~\ref{app:hm}), which is the condition required
by Corollary~\ref{cor:triad_confinement} to keep trajectories bounded
and the Kraus layer numerically inactive on all three test cases.

\noindent\textit{Notation.} Throughout this appendix, hats denote Fourier
coefficients, not operators. Operator hats are used exclusively in the
main text.

\subsection{Incompressible Navier--Stokes and Euler}
\label{app:ns}

The two-dimensional incompressible Navier--Stokes equations in
vorticity--streamfunction form are
\begin{equation}
\frac{\partial\omega}{\partial t} + J(\Psi, \omega)
= \nu\,\nabla^2 \omega, \qquad \omega = -\nabla^2 \Psi,
\label{eq:vorticity_ns}
\end{equation}
where $\omega$ is the vorticity, $\Psi$ is the streamfunction, and
$J(\Psi,\omega) = \Psi_x \omega_y - \Psi_y \omega_x$ is the Jacobian.
Incompressibility guarantees a vector potential; the relation
$\omega = -\nabla^2\Psi$ is the kinematic definition of vorticity.
The equations are posed on the doubly-periodic square
$\vect{x} \in [0, 2\pi]^2$ with periodic boundary conditions in both
spatial directions, admitting the Fourier expansion
$\omega(\vect{x},t) = \sum_{\vect{k}} \hat{\omega}_{\vect{k}}(t)\,
e^{\ii\vect{k}\cdot\vect{x}}$ over integer wavevectors $\vect{k} \in
\mathbb{Z}^2$. In Fourier space,
$\hat{\Psi}_{\vect{k}} = \hat{\omega}_{\vect{k}}/|\vect{k}|^2$.

For a resonant triad satisfying $\vect{k}_1 + \vect{k}_2 +
\vect{k}_3 = \vect{0}$, writing $\omega = \sum_{j=1}^{3} a_j(t)\,
e^{\ii\vect{k}_j \cdot \vect{x}} + \mathrm{c.c.}$\ (complex conjugate)
and projecting onto each wavevector gives~(\ref{eq:ns_triad}) with
\begin{equation}
C_j = (\vect{k}_m \times \vect{k}_n)\!\left(\frac{1}{|\vect{k}_m|^2}
- \frac{1}{|\vect{k}_n|^2}\right)\!, \quad (j,m,n)\;\text{cyclic}.
\label{eq:ns_coupling}
\end{equation}
The resonance condition implies that all three cross products are equal.
In two dimensions, $\vect{k}_m \times \vect{k}_n \equiv
k_{mx}k_{ny} - k_{my}k_{nx}$ is a scalar; denoting its common value by
$\kappa$:
\begin{equation}
\vect{k}_m \times \vect{k}_n = \kappa
\quad \text{for all cyclic } (j,m,n).
\label{eq:cross_product_identity}
\end{equation}
It follows that
\begin{align}
\sum_j C_j
&= \kappa\!\left[\left(\tfrac{1}{|\vect{k}_2|^2}
   - \tfrac{1}{|\vect{k}_3|^2}\right)
 + \left(\tfrac{1}{|\vect{k}_3|^2}
   - \tfrac{1}{|\vect{k}_1|^2}\right)\right.\nonumber\\
&\qquad\left.
 + \left(\tfrac{1}{|\vect{k}_1|^2}
   - \tfrac{1}{|\vect{k}_2|^2}\right)\right] = 0.
\label{eq:sum_cj_zero}
\end{align}
The cancellation is telescoping (each wavenumber enters the sum twice
with opposite signs) and follows directly from the Poisson inversion
required by incompressibility; the result holds for any resonant triad.

The energy conservation condition $\sum_j C_j/|\vect{k}_j|^2 = 0$ follows
by an identical telescoping argument:
\begin{align}
\sum_j \frac{C_j}{|\vect{k}_j|^2}
&= \kappa\!\left[
   \frac{1}{|\vect{k}_1|^2}\!\left(\frac{1}{|\vect{k}_2|^2}
   - \frac{1}{|\vect{k}_3|^2}\right)\right.\nonumber\\
&\qquad
 + \frac{1}{|\vect{k}_2|^2}\!\left(\frac{1}{|\vect{k}_3|^2}
   - \frac{1}{|\vect{k}_1|^2}\right)\nonumber\\
&\qquad\left.
 + \frac{1}{|\vect{k}_3|^2}\!\left(\frac{1}{|\vect{k}_1|^2}
   - \frac{1}{|\vect{k}_2|^2}\right)\right] = 0.
\label{eq:sum_cj_energy_zero}
\end{align}
Each parenthesised factor appears with opposite sign in the remaining
terms, producing exact cancellation. Physically, this expresses
conservation of modal kinetic energy
$E = \tfrac{1}{2}\sum_j a_j^2/|\vect{k}_j|^2$ by the nonlinear terms of
the incompressible Euler equations. Setting $\nu = 0$
gives~(\ref{eq:euler_triad}).

\subsection{Hasegawa--Mima equation}
\label{app:hm}

The Hasegawa--Mima equation~\cite{Hasegawa1978} governs electrostatic
drift-wave turbulence in magnetized plasmas:
\begin{equation}
\frac{\partial}{\partial t}(\nabla^2\phi - \phi) + [\phi,\, \nabla^2\phi]
= \mu\,\nabla^2(\nabla^2\phi - \phi),
\label{eq:hm}
\end{equation}
where $\phi$ is the electrostatic potential, $[\cdot,\cdot]$ denotes the
2D Poisson bracket, and $\mu$ is a collisional damping rate. The term
$-\phi$ arises from the ion polarization drift. The diamagnetic drift term
$v_*\partial_y\phi$ (where $v_*$ is the diamagnetic drift velocity) is
omitted here as it contributes only linear dispersion without altering the
triad interaction structure or the conserved quantities. As with the
Navier--Stokes case, the equation is posed on the doubly-periodic square
$\vect{x} \in [0, 2\pi]^2$ with periodic boundary conditions in both
spatial directions, admitting a Fourier expansion
$\phi(\vect{x},t) = \sum_{\vect{k}} \hat{\phi}_{\vect{k}}(t)\,
e^{\ii\vect{k}\cdot\vect{x}}$ over integer wavevectors $\vect{k} \in
\mathbb{Z}^2$.

Defining the modified vorticity $q = (\nabla^2 - 1)\phi$ and taking the
Fourier transform gives $\hat{q}_{\vect{k}} = -(|\vect{k}|^2+1)\,
\hat{\phi}_{\vect{k}}$, so the Poisson inversion reads
\begin{equation}
\hat{\phi}_{\vect{k}} =
-\frac{\hat{q}_{\vect{k}}}{|\vect{k}|^2+1},
\label{eq:hm_inversion}
\end{equation}
replacing the $1/|\vect{k}|^2$ of the NS case by
$-(|\vect{k}|^2+1)^{-1}$. Writing the modified vorticity in the resonant
triad as $q(\vect{x},t) = \sum_{j=1}^{3} b_j(t)\,
e^{\ii\vect{k}_j\cdot\vect{x}} + \mathrm{c.c.}$, the modal amplitudes
$b_j$ are the Fourier coefficients of $q$, and the corresponding
potential amplitudes follow from~(\ref{eq:hm_inversion}) as
$\phi_j = -b_j/(|\vect{k}_j|^2 + 1)$. Projecting the Poisson
bracket~(\ref{eq:hm}) onto each wavevector gives~(\ref{eq:hm_triad})
with coupling coefficient
\begin{align}
C_j^{\mathrm{HM}} = (\vect{k}_m \times \vect{k}_n)
\left(\frac{1}{|\vect{k}_m|^2 + 1}
- \frac{1}{|\vect{k}_n|^2 + 1}\right)\!.
\label{eq:hm_coupling}
\end{align}
This differs from the NS coupling~(\ref{eq:ns_coupling}) only by the
replacement $|\vect{k}|^{-2} \to (|\vect{k}|^2+1)^{-1}$. The difference
structure is identical, so the telescoping identity~(\ref{eq:sum_cj_zero})
carries over, giving $\sum_j C_j^{\mathrm{HM}} = 0$. The energy
conservation condition $\sum_j C_j^{\mathrm{HM}}/(|\vect{k}_j|^2+1) = 0$
follows from the same argument with $|\vect{k}|^{-2}$ replaced throughout.
In the inviscid limit ($\mu = 0$), the equation therefore conserves both
$Z = \tfrac{1}{2}\sum_j b_j^2$ and the energy
$E = \tfrac{1}{2}\sum_j b_j^2/(|\vect{k}_j|^2 + 1)$.

\bibliography{kvn}

\end{document}